\documentclass[11pt,a4paper]{article}

\usepackage[utf8]{inputenc}
\usepackage[expansion=false]{microtype}
\usepackage{amsmath,amssymb,amsthm}
\usepackage{mathtools}
\usepackage{bm}
\usepackage{graphicx}
\usepackage{xcolor}
\usepackage{booktabs}
\usepackage{array}
\usepackage{multirow}
\usepackage{xurl}
\usepackage{hyperref}
\usepackage[capitalise,noabbrev]{cleveref}
\usepackage{algorithm}
\usepackage{algpseudocode}
\usepackage{geometry}
\usepackage{enumitem}
\usepackage{longtable}
\usepackage[numbers,sort&compress]{natbib}
\usepackage{tikz}
\usetikzlibrary{arrows.meta, positioning, shapes.geometric, fit, backgrounds, calc}

\setlist{itemsep=2pt, topsep=4pt}
\hypersetup{
  colorlinks=true,
  linkcolor=blue!70!black,
  citecolor=green!50!black,
  urlcolor=blue!70!black,
  pdftitle={An Information-Geometric Framework for OCS-Reconfigurable Interconnect Topology in Large-Scale AI Clusters},
  pdfauthor={Shintaro Yoshizawa}
}

\theoremstyle{plain}
\newtheorem{theorem}{Theorem}[section]
\newtheorem{proposition}[theorem]{Proposition}
\newtheorem{lemma}[theorem]{Lemma}
\newtheorem{corollary}[theorem]{Corollary}
\theoremstyle{definition}
\newtheorem{definition}[theorem]{Definition}
\newtheorem{assumption}{Assumption}
\crefname{assumption}{Assumption}{Assumptions}

\newtheorem{remark}[theorem]{Remark}

\newcommand{\calM}{\mathcal{M}}
\newcommand{\calW}{\mathcal{W}}
\newcommand{\calS}{\mathcal{S}}
\newcommand{\calG}{\mathcal{G}}
\newcommand{\calP}{\mathcal{P}}
\newcommand{\calA}{\mathcal{A}}
\newcommand{\calC}{\mathcal{C}}
\newcommand{\bfv}{\mathbf{v}}
\newcommand{\bfg}{\mathbf{g}}
\newcommand{\bfb}{\mathbf{b}}
\newcommand{\bfw}{\mathbf{w}}
\newcommand{\bfq}{\mathbf{q}}
\newcommand{\bfx}{\mathbf{x}}
\newcommand{\bfy}{\mathbf{y}}
\newcommand{\bfd}{\mathbf{d}}
\newcommand{\bfP}{\mathbf{P}}
\newcommand{\bfeta}{\boldsymbol{\eta}}

\newcommand{\bfrho}{\boldsymbol{\rho}}
\newcommand{\one}{\mathbf{1}}
\newcommand{\DKL}{D_{\mathrm{KL}}}
\newcommand{\KL}[2]{\DKL(#1\,\|\,#2)}
\newcommand{\norm}[1]{\left\|#1\right\|}
\newcommand{\inner}[2]{\langle #1,\, #2 \rangle}
\newcommand{\R}{\mathbb{R}}
\newcommand{\E}{\mathbb{E}}
\newcommand{\Simp}{\Delta}
\newcommand{\Eocs}{E}
\newcommand{\qW}{\bfq_{\calW}}
\newcommand{\lse}{\operatorname{LSE}}
\DeclareMathOperator*{\argmin}{arg\,min}

\newcommand{\Jstatic}{2.204}
\newcommand{\Jadapt}{1.384}
\newcommand{\Jgain}{37}
\newcommand{\DynOracle}{853.7}
\newcommand{\DynT}{600}
\newcommand{\DynVone}{4.7}
\newcommand{\BestPer}{142.1}
\newcommand{\BestTV}{116.2}
\newcommand{\BestGLR}{103.2}
\newcommand{\NeverRegret}{1282.0}
\newcommand{\StaticRegret}{525.0}
\newcommand{\MTref}{1.119 / 0.771 / 0.659}

\begin{document}

\title{%
  \textbf{An Information-Geometric Framework for\\
  OCS-Reconfigurable Interconnect Topology\\
  in Large-Scale AI Clusters}
}

\author{%
  Shintaro Yoshizawa\\[2pt]
  {\normalsize Nagoya Mathematical and Information Science Research}\\[1pt]
  {\normalsize\texttt{\href{mailto:shintaro.yoshizawa.net@gmail.com}{shintaro.yoshizawa.net@gmail.com}}}
}

\date{}

\maketitle
\thispagestyle{empty}

\begin{abstract}
The interconnects of large AI training clusters, exemplified by NVIDIA's GB200 NVL72
systems, are built from static and symmetric fabrics: an NVLink domain within the rack and a
Fat-Tree between racks.  Optical circuit switches render the inter-pod wiring
programmable, and thereby raise three design questions that static fabrics never pose:
which wiring to deploy, how far it should depart from the symmetric reference, and when it
should be changed.  This paper answers them within a single information-geometric
framework.

A wiring is a point of the probability simplex over candidate circuits, constrained to a
degree polytope, and its departure from a symmetric prior encoding the uniform logical
capacity of a non-blocking Fat-Tree is measured by the Kullback--Leibler divergence.  The
cost is a smoothed critical-path length of the training-iteration DAG; it is convex with
explicit smoothness constants, and its gradient is given by the marginals of a Gibbs
distribution over DAG paths, computable by a forward--backward recursion.  For the free
energy $F_{\beta_1}=U+\beta_1^{-1}\DKL(\cdot\,\|\,\bfq)$ we establish well-posedness, monotone
dependence on the inverse temperature $\beta_1$ with quantitative limits at both ends, a lower
bound on the gain of adaptation by the Jeffreys divergence between the static and the
adaptive wiring, and a symmetry principle that characterises exactly when static wiring is
optimal.  The associated entropic proximal-gradient method reduces each step to a
multiplicative update followed by a KL projection with a $|V|$-dimensional dual, and converges
linearly in KL divergence, also in the presence of inexact projections.

For time-varying traffic, a threshold policy attains regret $O(\sqrt{\Gamma CV_1T})$ against
the switching-free dynamic optimum, with a finite-sample counterpart whose radius depends on
the support size rather than the ambient dimension; the KL indicator of traffic change is
identified as the multinomial generalised likelihood-ratio statistic, which yields an
adaptive-lag GLR--CUSUM detector with a non-asymptotic false-alarm bound.  For clusters shared
by several jobs, generalised-mean fairness of slowdowns is convex, coincides in its maximisers
with $\alpha$-fairness of throughputs for $\alpha\ge2$, and leads to a Gibbs-form optimality
condition.  Synthetic experiments on pipeline-, data- and expert-parallel workloads confirm
the theory and delineate its scope, including the regime in which exact integer programming
remains preferable.
\end{abstract}

\vspace{2mm}

\newpage
\setcounter{page}{1}
\tableofcontents
\newpage

\section{Introduction}
\label{sec:introduction}

\subsection{The Communication Bottleneck in LLM Training}

Large language models (LLMs) now reach hundreds of billions to trillions of parameters
\cite{brown2020gpt3,chowdhery2022palm,touvron2023llama}, and training them requires tens of
thousands of accelerators operating in concert.  The interconnect is a first-order
determinant of throughput and of reliability: during a 54-day snapshot of Llama~3 405B
pre-training, 35 of 419 unexpected interruptions (8.4\%) were attributed to network switches
and cables \cite{metaLlama3}.

A training job is decomposed along tensor (TP), pipeline (PP), data (DP) and, for
mixture-of-experts (MoE) models, expert (EP) parallelism \cite{shoeybi2019megatron,huang2019gpipe}.
Each dimension induces a distinct collective: TP issues latency-critical collectives that
normally stay inside the scale-up domain, PP sends activations and gradients between
adjacent stages, DP performs gradient all-reduce, and EP performs all-to-all exchanges
whose volume depends on token routing.  The resulting traffic is strongly structured,
phase-dependent, and---for MoE---time-varying.

\subsection{Fat-Tree as the Dominant Paradigm and Its Limitations}

The $k$-ary Fat-Tree \cite{leiserson1985fattrees,alfares2008scalable} is the de facto standard
for HPC and AI clusters.  NVIDIA's DGX SuperPOD reference architecture uses a multi-tier
InfiniBand Fat-Tree \cite{nvidiadgx}; the Blackwell generation adds a rack-scale NVLink
domain (GB200 NVL72: 72 GPUs, fifth-generation NVLink, 1.8\,TB/s per GPU and 130\,TB/s in
aggregate) with an InfiniBand or Ethernet Fat-Tree between racks \cite{nvidiagb200}.  A
non-blocking Fat-Tree offers the same logical capacity to every host pair and is therefore
optimal for worst-case (uniform) traffic.  For structured LLM traffic it is
over-provisioned in some directions and cannot be re-shaped when the parallelism layout
changes: the only remaining degree of freedom is routing, constrained by fixed wiring.

\subsection{Optical Circuit Switching: Promise and Open Problems}

Optical circuit switches establish reconfigurable point-to-point light paths that bypass
electrical packet switching.  Google's TPU\,v4 \cite{jouppi2023tpuv4} interconnects 4{,}096
chips through a 3D torus whose slices are composed per job by 48 OCS units, and Google's
data-centre networks use OCS at scale \cite{urata2022apollo,poutievski2022jupiter}.  Topology
engineering for such fabrics is typically posed as an integer or mixed-integer programme:
COUDER \cite{teh2020couder} optimises over sets of traffic matrices and reports, on Facebook
traces, about 20\% higher throughput and 32\% lower average hop count than cost-equivalent
static topologies; ATRO \cite{atro2025} jointly optimises topology and routing, finding the
global optimum for one-hop routing and alternating updates for multi-hop routing;
TopoOpt \cite{wang2023topoopt} co-optimises topology and parallelisation for training jobs;
DELTA \cite{delta2025} encodes the training DAG in a MILP and reports up to 17.5\% lower
communication time and at least 20\% fewer optical ports than traffic-matrix baselines; and
InfiniteHBD \cite{infhbd2025} builds a data-centre-scale high-bandwidth domain from OCS
transceivers, reporting a cost of 31\% of NVL72.
What is missing is a \emph{continuous} theory that (a) states what a reasonable departure
from the symmetric Fat-Tree reference is, (b) admits provable, low-complexity updates,
(c) incorporates the causal structure of the training DAG, and (d) gives a statistically
principled rule for \emph{when} to reconfigure.

\subsection{Contributions and Their Relation to Known Results}

We develop such a theory by combining elementary information geometry
\cite{amari2000methods,amari2016ig,csiszar1975idivergence} with first-order methods in Bregman
geometry \cite{bauschke2017descent,lu2018relatively} and classical sequential analysis
\cite{lorden1971procedures,siegmund1995glr}.  The tools are standard; the contribution lies in
the formulation they make possible and in the results specific to it.  We regard the
following three as the principal contributions.

\begin{enumerate}[label=(\Roman*)]
\item \textbf{A temperature-controlled design principle (\cref{sec:freeenergy}).}  Wirings are
  points of a simplex restricted to a degree polytope, and their departure from a symmetric
  Fat-Tree-like prior is penalised by the KL divergence with weight $\beta_1^{-1}$.  We prove
  well-posedness and strict positivity (\cref{thm:wellposed}), monotone and quantitative limits
  in $\beta_1$ (\cref{thm:limits}), a lower bound on the gain of adaptation by the Jeffreys
  divergence between static and adaptive wirings (\cref{thm:jeffreys}), and a symmetry
  principle identifying exactly when static wiring is optimal
  (\cref{prop:symmetry,prop:staticopt}).
\item \textbf{A differentiable DAG cost (\cref{sec:model}).}  The iteration time is a smoothed
  critical path of the training DAG; it is convex with explicit smoothness constants
  (\cref{lem:smooth}), and its gradient is the vector of path criticalities of a Gibbs
  distribution over DAG paths, computed by a forward--backward recursion
  (\cref{prop:criticality}).
\item \textbf{Guarantees for reconfiguration (\cref{sec:reconf}).}  A threshold policy has
  regret $O(\sqrt{\Gamma CV_1T})$ against the switching-free dynamic optimum
  (\cref{thm:regret}), with a KL-threshold variant (\cref{cor:kl}) and a finite-sample version
  whose radius scales with the support size (\cref{thm:regret_sampled}); the GLR--CUSUM detector
  obtained from the likelihood-ratio form of the KL indicator controls false alarms
  (\cref{thm:falsealarm}).
\end{enumerate}

Two further components apply established results to the present model and are included to
make the framework complete and verifiable: the convergence analysis of the entropic
proximal-gradient method (\cref{sec:algorithm}; \cref{thm:convergence,thm:inexact}), which
specialises the theory of Bregman proximal gradient methods under relative smoothness, and the
multi-tenant extension (\cref{sec:multitenant}), which combines the KL chain rule with
generalised-mean fairness and identifies the range $\alpha\ge2$ in which fairness and convexity
are compatible (\cref{prop:fairconvex,thm:kkt,prop:pareto}).  Synthetic experiments
(\cref{sec:experiments}) verify every bound and compare against static, demand-proportional and
exact integer baselines.  We have not carried out an exhaustive priority search, and make no
claim that the individual inequalities are new as abstract statements; their interest lies in
what they assert about reconfigurable fabrics.

\section{Background}
\label{sec:background}

\subsection{Fat-Tree, Blackwell, and OCS Fabrics}
\label{sec:fattree}

A \emph{$k$-ary Fat-Tree} has $(k/2)^2$ core switches and $k$ pods, each with $k/2$
aggregation and $k/2$ edge switches; every edge switch serves $k/2$ hosts, for $k^3/4$ hosts
and $k^3/2$ inter-switch links in total.  Its automorphism group permutes pods and the hosts
below an edge switch, and with full bisection bandwidth $(k^3/4)(c/2)$ it offers the same
logical capacity to every host pair.  \Cref{fig:topology_comparison}(a) shows $k=4$.

In Blackwell NVL72 systems \cite{nvidiagb200} the 72 GPUs of a rack are connected through
NVSwitch trays; each GPU also has a NIC to the scale-out InfiniBand or Ethernet Fat-Tree
(\cref{fig:topology_comparison}(b)).  All layers are static.

An OCS device establishes point-to-point optical circuits between transceivers; switching
takes milliseconds for MEMS devices, while end-to-end reconfiguration including transceiver
and NIC re-initialisation is reported in the range of seconds \cite{patronas2025optical,delta2025}.
We work at \emph{pod} granularity: a pod is a scale-up domain (e.g.\ a rack) with $r$ OCS
uplinks; the OCS layer realises any multigraph on the pods in which every pod has degree at
most $r$ and the total number of circuits does not exceed the OCS port budget $\mu$
(\cref{fig:topology_comparison}(c)).

\begin{figure}[!tbp]
\centering
\tikzset{
  host/.style    = {circle, draw=black!70, fill=black!15, minimum size=4.2pt, inner sep=0pt},
  swbase/.style  = {rectangle, rounded corners=1.5pt, draw=black!70,
                    minimum width=13pt, minimum height=8pt, inner sep=1pt, font=\tiny},
  core/.style    = {swbase, fill=blue!18},
  agg/.style     = {swbase, fill=green!18},
  edgesw/.style  = {swbase, fill=yellow!30},
  nvsw/.style    = {swbase, fill=white, minimum width=40pt},
  lk/.style      = {draw=black!55, line width=0.45pt},
  nvlk/.style    = {draw=blue!70, line width=0.7pt},
  iblk/.style    = {draw=red!65, line width=0.45pt},
  tier/.style    = {font=\scriptsize, text=black!60},
  ocsnode/.style = {circle, draw=orange!80!black, fill=orange!20, minimum size=8pt, inner sep=0pt},
  torlk/.style   = {draw=orange!80!black, line width=0.7pt, dashed},
  circ/.style    = {draw=orange!90!black, line width=1.6pt},
  fib/.style     = {draw=black!45, line width=0.4pt},
  pt/.style      = {circle, minimum size=5.5pt, inner sep=0pt, draw=white, line width=0.4pt},
  ptlab/.style   = {font=\scriptsize, inner sep=1pt},
}
\begin{minipage}[t]{0.49\linewidth}
\centering
{\small\bfseries (a) Fat-Tree ($k{=}4$)}\\[-1pt]
{\scriptsize\color{black!60} Static, symmetric}\\[4pt]
\begin{tikzpicture}
  \foreach \p in {0,1,2,3}{
    \pgfmathsetmacro\cx{1.5+1.6*\p}
    \node[agg]    (a\p0) at (\cx-0.4,2.05) {A};
    \node[agg]    (a\p1) at (\cx+0.4,2.05) {A};
    \node[edgesw] (e\p0) at (\cx-0.4,1.15) {E};
    \node[edgesw] (e\p1) at (\cx+0.4,1.15) {E};
    \foreach \i in {0,1}{
      \foreach \j in {0,1}{ \draw[lk] (a\p\i) -- (e\p\j); }
      \pgfmathsetmacro\ex{\cx-0.4+0.8*\i}
      \node[host] (h\p\i a) at (\ex-0.17,0.3) {};
      \node[host] (h\p\i b) at (\ex+0.17,0.3) {};
      \draw[lk] (e\p\i) -- (h\p\i a); \draw[lk] (e\p\i) -- (h\p\i b);
    }
    \begin{scope}[on background layer]
      \fill[black!4, rounded corners=3pt] (\cx-0.78,0.08) rectangle (\cx+0.78,2.32);
    \end{scope}
  }
  \foreach \c/\cx in {0/2.2,1/3.5,2/4.9,3/6.2}{
    \node[core] (c\c) at (\cx,3.05) {C};
  }
  \foreach \c in {0,1,2,3}{
    \pgfmathtruncatemacro\g{\c/2}
    \foreach \p in {0,1,2,3}{ \draw[lk] (c\c) -- (a\p\g); }
  }
  \node[tier, anchor=east] at (0.62,3.05) {Core};
  \node[tier, anchor=east] at (0.62,2.05) {Agg};
  \node[tier, anchor=east] at (0.62,1.15) {Edge};
  \node[tier, anchor=east] at (0.62,0.30) {Host};
  \foreach \p in {0,1,2,3}{
    \node[tier, font=\tiny] at (1.5+1.6*\p,-0.12) {pod \p};
  }
\end{tikzpicture}
\end{minipage}\hfill
\begin{minipage}[t]{0.49\linewidth}
\centering
{\small\bfseries (b) Blackwell NVL72}\\[-1pt]
{\scriptsize\color{black!60} NVLink (intra-rack) + InfiniBand Fat-Tree (inter-rack)}\\[4pt]
\begin{tikzpicture}
  \node[core] (s1) at (2.6,3.05) {Spine};
  \node[core] (s2) at (4.4,3.05) {Spine};
  \foreach \l/\lx in {1/1.2,2/2.6,3/4.4,4/5.8}{
    \node[agg] (l\l) at (\lx,2.2) {Leaf};
    \draw[iblk] (s1) -- (l\l); \draw[iblk] (s2) -- (l\l);
  }
  \foreach \r/\rx/\la/\lb in {0/1.9/1/2, 1/5.1/3/4}{
    \begin{scope}[on background layer]
      \fill[blue!5, draw=blue!30, rounded corners=3pt]
        (\rx-1.45,0.0) rectangle (\rx+1.45,1.62);
    \end{scope}
    \node[nvsw] (nv\r) at (\rx,0.28) {NVSwitch};
    \foreach \g in {0,...,5}{
      \pgfmathsetmacro\gx{\rx-1.15+0.46*\g}
      \node[host] (g\r\g) at (\gx,1.15) {};
      \draw[nvlk] (g\r\g) -- (nv\r.north);
    }
    \node[font=\tiny, text=black!60] at (\rx+1.18,1.15) {$\cdots$};
    \node[font=\tiny, text=blue!60!black, anchor=north] at (\rx,-0.03) {NVL72 rack};
    \foreach \g in {0,1,2}{ \draw[iblk] (g\r\g) -- (l\la); }
    \foreach \g in {3,4,5}{ \draw[iblk] (g\r\g) -- (l\lb); }
  }
  \node[tier, anchor=west] at (6.55,3.05) {IB};
  \node[tier, anchor=west] at (6.55,2.2)  {IB};
  \node[tier, anchor=west] at (6.55,1.15) {GPU};
  \draw[nvlk] (1.6,-0.6) -- (2.1,-0.6); \node[tier, anchor=west] at (2.12,-0.6) {NVLink};
  \draw[iblk] (3.7,-0.6) -- (4.2,-0.6); \node[tier, anchor=west] at (4.22,-0.6) {InfiniBand};
\end{tikzpicture}
\end{minipage}

\vspace{10pt}
\begin{minipage}[t]{0.49\linewidth}
\centering
{\small\bfseries (c) OCS-Reconfigurable}\\[-1pt]
{\scriptsize\color{black!60} Logical 2D/3D torus, re-wired per job}\\[4pt]
\begin{tikzpicture}
  \node[draw=orange!80!black, fill=orange!12, rounded corners=3pt,
        minimum width=4.6cm, minimum height=0.7cm] (ocs) at (2.2,4.25) {};
  \node[font=\scriptsize, text=orange!50!black, anchor=west] at (ocs.west) {\;OCS};
  \foreach \r in {0,1,2}{
    \foreach \c in {0,1,2}{
      \node[ocsnode] (n\r\c) at (1.0+1.2*\c, 0.45+1.05*\r) {};
    }
  }
  \foreach \r in {0,1,2}{
    \draw[torlk] (n\r0) -- (n\r1); \draw[torlk] (n\r1) -- (n\r2);
    \draw[torlk, -{Stealth[length=3pt]}] (n\r2) -- ++(0.55,0);
    \draw[torlk, {Stealth[length=3pt]}-] (n\r0) -- ++(-0.55,0);
  }
  \foreach \c in {0,1,2}{
    \draw[torlk] (n0\c) -- (n1\c); \draw[torlk] (n1\c) -- (n2\c);
    \draw[torlk, {Stealth[length=3pt]}-] (n0\c) -- ++(0,-0.45);
  }
  \foreach \c in {0,1,2}{
    \draw[fib] (n2\c) -- (n2\c |- ocs.south);
  }
  \draw[circ] (n20.north) -- (n20 |- ocs.south)
              to[out=90, in=90, looseness=1.1] (n22 |- ocs.south) -- (n22.north);
  \draw[torlk] (0.2,-0.35) -- (0.7,-0.35);
  \node[tier, anchor=west] at (0.72,-0.35) {torus link};
  \draw[circ] (2.3,-0.35) -- (2.8,-0.35);
  \node[tier, anchor=west] at (2.82,-0.35) {OCS circuit};
\end{tikzpicture}
\end{minipage}\hfill
\begin{minipage}[t]{0.49\linewidth}
\centering
{\small\bfseries (d) Topology manifold $\calM_{\mathrm{topo}}$}\\[-1pt]
{\scriptsize\color{black!60} Information-geometric view}\\[4pt]
\begin{tikzpicture}
  \fill[blue!6, draw=blue!35, rounded corners=16pt] (0,0) rectangle (6.4,4.5);
  \node[font=\scriptsize, text=blue!50!black, anchor=north west] at (0.3,4.4) {$\calM_{\mathrm{topo}}$};
  \filldraw[fill=green!10, draw=green!50!black, dashed, line width=0.6pt]
    (0.7,0.55) -- (4.3,0.4) -- (5.0,2.3) -- (3.6,3.75) -- (0.9,3.3) -- cycle;
  \node[font=\scriptsize, text=green!40!black] at (1.25,2.95) {$\calW$};
  \coordinate (w0) at (5.55,3.45);
  \draw[orange!45, dashed] (w0) ellipse (0.32 and 0.22);
  \draw[orange!35, dashed] (w0) ellipse (0.72 and 0.5);
  \node[pt, fill=orange!90!black] at (w0) {};
  \node[ptlab, anchor=south] at ($(w0)+(0,0.22)$) {$\bfw^{*}_{0}$};
  \coordinate (wh) at (4.45,2.85);
  \draw[->, gray, dashed, line width=0.6pt, shorten >=3pt, shorten <=3pt] (w0) -- (wh);
  \node[ptlab, text=black!70, anchor=north west, font=\tiny] at (5.1,2.85) {KL proj.};
  \node[pt, fill=orange!55] at (wh) {};
  \node[ptlab, anchor=east] at ($(wh)+(-0.1,0.06)$) {$\hat{\bfw}$};
  \coordinate (wft) at (0.7,0.55);
  \node[pt, fill=blue!65] at (wft) {};
  \node[pt, fill=red!60, minimum size=4pt] at ($(wft)+(0.11,0.06)$) {};
  \node[ptlab, anchor=north west, align=left] at ($(wft)+(0.05,-0.14)$)
     {$\qW$ (static, FT/BW)};
  \node[pt, fill=purple!60] (df) at (1.6,1.9) {};
  \node[ptlab, anchor=south] at ($(df)+(0,0.08)$) {Dragonfly+};
  \node[pt, fill=teal!60] (jf) at (2.3,1.2) {};
  \node[ptlab, anchor=south] at ($(jf)+(0,0.08)$) {Jellyfish};
  \node[pt, fill=violet!50] (hm) at (2.9,2.6) {};
  \node[ptlab, anchor=south] at ($(hm)+(0,0.08)$) {HammingMesh};
  \draw[-{Stealth[length=5pt]}, orange!80!black, line width=1pt, shorten <=4pt, shorten >=4pt]
    (wft) .. controls (2.6,0.55) and (3.8,1.4) .. (wh);
  \node[font=\scriptsize, text=orange!70!black] at (5.55,1.25) {proximal steps};
\end{tikzpicture}
\end{minipage}

\caption{%
  \textbf{Four views of interconnect topology.}
  \textbf{(a)} A $k{=}4$ Fat-Tree with 4 pods, $(k/2)^2{=}4$ core switches, 8 aggregation
  switches, 8 edge switches, and $k^3/4{=}16$ hosts.  All links are fixed at deployment.
  \textbf{(b)} NVIDIA Blackwell NVL72: GPUs are fully meshed within each rack via NVSwitch
  (blue links), and racks are interconnected through a hierarchical InfiniBand Fat-Tree reached
  through per-GPU NICs (red links).  Only 6 of the 72 GPUs per rack are drawn.  The topology is static.
  \textbf{(c)} An OCS-reconfigurable fabric realizes a logical torus (dashed orange; arrows
  indicate wrap-around links) whose wiring is reprogrammed per job.  The thick orange
  circuit is an example cross-connect set inside the OCS, turning otherwise idle fibres into a
  long-range link.
  \textbf{(d)} Schematic information-geometric view (not to scale): wirings are points of the
  simplex $\calM_{\mathrm{topo}}$, and the degree polytope $\calW$ is $m$-convex.  The static
  symmetric wiring sits at the prior; an iterate $\bfw^{*}_{0}$ outside $\calW$ (e.g.\ the
  multiplicative step \eqref{eq:geomix}) is mapped back by the KL projection
  $\hat{\bfw}=\argmin_{\calW}\KL{\cdot}{\bfw^{*}_0}$ (\cref{lem:closedform}).  Other topology
  families are shown only for orientation; families with different switch sets are not, in
  general, points of the same simplex (\cref{rem:support}).
}
\label{fig:topology_comparison}
\end{figure}
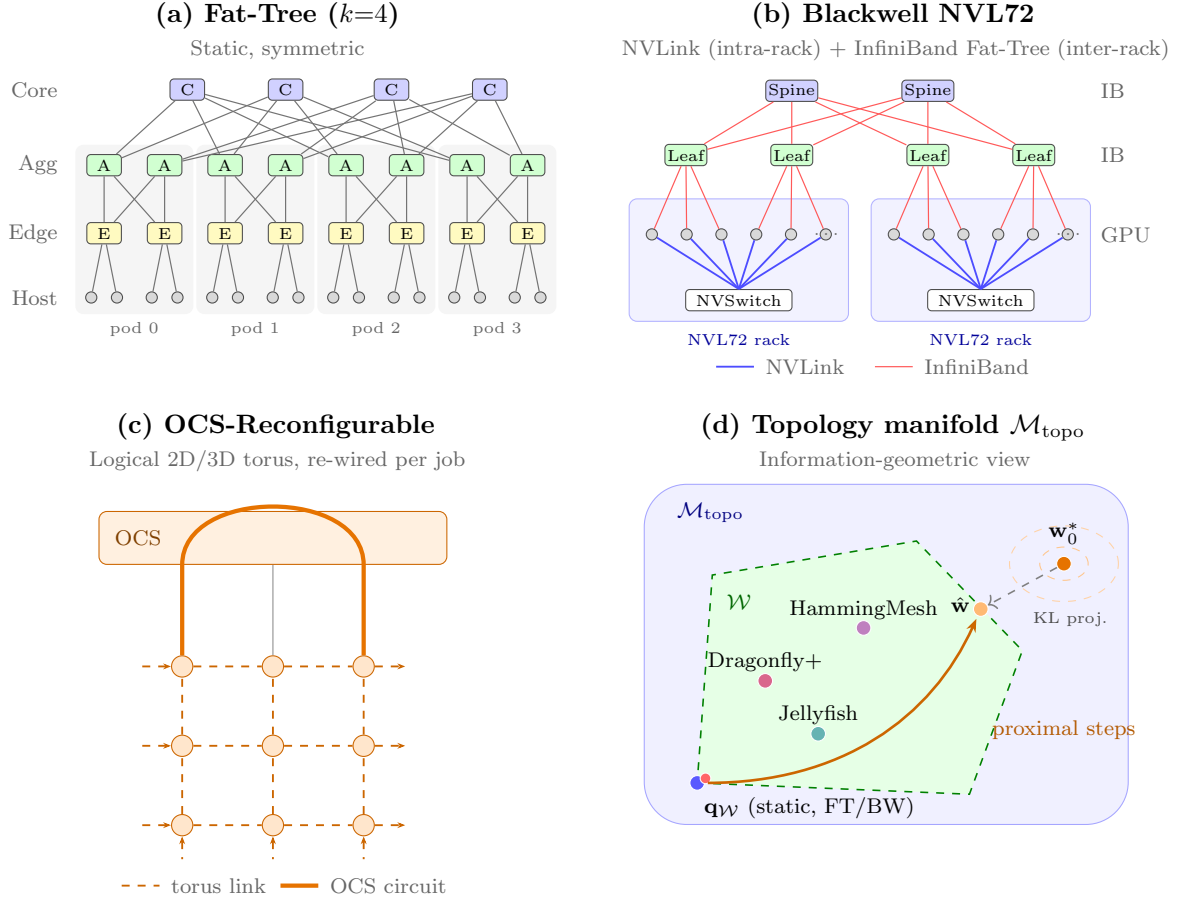

\subsection{Information Geometry of the Simplex}
\label{sec:ig_background}

Let $\Eocs$ be a finite set and
$\Simp=\{\bfw\in\R^{\Eocs}_{\ge0}:\sum_e w_e=1\}$, with relative interior $\Simp^\circ$.
With the Fisher metric $g_{ij}=\sum_e p_e^{-1}\partial_ip_e\,\partial_jp_e$, $\Simp^\circ$ is a
dually flat manifold \cite{amari2000methods,amari2016ig}: the $m$-affine coordinates are
$\bfeta=\bfw$ and the $e$-affine coordinates are $\theta_e=\log w_e-\log w_{e_0}$.  The
canonical divergence is the KL divergence
$\KL{\bfx}{\bfw}=\sum_e x_e\log(x_e/w_e)$ (with $0\log0=0$), which is the Bregman
divergence of the negative entropy $h(\bfw)=\sum_e w_e\log w_e$ restricted to $\Simp$.

A set $\calW\subseteq\Simp$ is \emph{$m$-convex} if it is convex in $\bfeta$-coordinates
(e.g.\ a polytope cut out by linear inequalities in $\bfw$), and \emph{$m$-flat} if it is the
intersection of $\Simp$ with an affine subspace.  For a closed $m$-convex $\calW$ and
$\bfq\in\Simp^\circ$ we call
\begin{equation}
  \qW=\argmin_{\bfw\in\calW}\KL{\bfw}{\bfq}
  \label{eq:iproj}
\end{equation}
the \emph{KL projection} (Csisz\'ar's I-projection \cite{csiszar1975idivergence}) of $\bfq$
onto $\calW$.  In Amari's terminology the minimiser over the first argument onto an
$m$-flat set is reached along an $e$-geodesic and is called the $e$-projection; to avoid
the ambiguity present in parts of the literature we use the explicit form \eqref{eq:iproj}
throughout.

\begin{lemma}[Positivity]
\label{lem:positivity}
Let $\calC\subseteq\Simp$ be closed and convex and contain a point with all coordinates
positive.  Let $\varphi(\bfw)=c\,h(\bfw)+\psi(\bfw)$ with $c>0$ and $\psi$ convex and
continuously differentiable on a neighbourhood of $\Simp$.  Then every minimiser of
$\varphi$ over $\calC$ has all coordinates positive.
\end{lemma}

\begin{proof}
Let $\bar\bfw$ minimise $\varphi$ over $\calC$ and suppose $\bar w_e=0$ for some $e$.  Pick
$\bfx\in\calC$ with $\bfx>0$ and set $\bfw_t=\bar\bfw+t(\bfx-\bar\bfw)\in\calC$ for
$t\in(0,1]$.  For coordinates with $\bar w_f>0$, the map $t\mapsto w_{t,f}\log w_{t,f}$ is
differentiable at $0$; so is $\psi(\bfw_t)$.  For the coordinate $e$,
$\frac{d}{dt}\,(t x_e)\log(t x_e)=x_e(\log(tx_e)+1)\to-\infty$ as $t\downarrow0$.  Hence the
right derivative of $t\mapsto\varphi(\bfw_t)$ at $0$ is $-\infty$, so
$\varphi(\bfw_t)<\varphi(\bar\bfw)$ for small $t>0$, a contradiction.
\end{proof}

\begin{theorem}[Generalised Pythagorean inequality {\cite{csiszar1975idivergence}}]
\label{thm:pythagoras}
Let $\calW\subseteq\Simp$ be closed, $m$-convex and contain a positive point, and let
$\bfq\in\Simp^\circ$.  Then $\qW$ exists, is unique, is positive, and for every
$\bfx\in\calW$,
\begin{equation}
  \KL{\bfx}{\bfq}\;\ge\;\KL{\bfx}{\qW}+\KL{\qW}{\bfq},
  \label{eq:pythagoras}
\end{equation}
with equality for all $\bfx\in\calW$ when $\calW$ is $m$-flat.
\end{theorem}

\begin{proof}
$\bfw\mapsto\KL{\bfw}{\bfq}=h(\bfw)-\inner{\log\bfq}{\bfw}$ is continuous and strictly
convex on the compact set $\calW$, so the minimiser exists and is unique; it is positive by
\cref{lem:positivity}.  At a positive minimiser the first-order condition reads
$\inner{\log(\qW/\bfq)+\one}{\bfx-\qW}\ge0$ for all $\bfx\in\calW$; since
$\inner{\one}{\bfx-\qW}=0$ this is $\sum_e(x_e-q_{\calW,e})\log(q_{\calW,e}/q_e)\ge0$.
Direct expansion gives the identity
\begin{equation}
  \KL{\bfx}{\bfq}-\KL{\bfx}{\qW}-\KL{\qW}{\bfq}
  =\sum_e(x_e-q_{\calW,e})\log\frac{q_{\calW,e}}{q_e},
  \label{eq:three_point_identity}
\end{equation}
which proves \eqref{eq:pythagoras}.  If $\calW$ is $m$-flat, $\qW$ lies in its relative
interior (it is positive), so $\qW\pm s(\bfx-\qW)\in\calW$ for small $s>0$ and the
first-order condition holds with equality.
\end{proof}

We repeatedly use \emph{Pinsker's inequality} $\KL{\bfx}{\bfw}\ge\frac12\norm{\bfx-\bfw}_1^2$
\cite{cover2006elements} and the following three-point identity, valid for
$\bfx\in\Simp$ and $\bfw,\bfq\in\Simp^\circ$, which is \eqref{eq:three_point_identity} with
$\qW$ replaced by any positive $\bfw$:
\begin{equation}
  \KL{\bfx}{\bfq}=\KL{\bfw}{\bfq}+\inner{\log(\bfw/\bfq)}{\bfx-\bfw}+\KL{\bfx}{\bfw}.
  \label{eq:three_point}
\end{equation}

\section{Model: Wiring Manifold, Prior, and DAG Cost}
\label{sec:model}

\subsection{Wirings and the Degree Polytope}
\label{sec:wiring}

Let $V$ be the set of pods, $|V|=N$, each with $r$ OCS uplinks of bandwidth $c$, and let
$\mu\le Nr/2$ be the number of circuits the OCS layer can establish (its port budget).  Let
$\Eocs$ be the set of candidate pod pairs (in the experiments all $\binom N2$ pairs).

\begin{definition}[Wiring, capacity, degree polytope]
A \emph{wiring} is $\bfw\in\Simp$; it places $\mu w_e$ circuits on pair $e$.  The
capacity of pair $e$ is
\begin{equation}
  \kappa_e(\bfw)=b_0+B\,w_e,\qquad B=\mu c,
  \label{eq:capacity}
\end{equation}
where $b_0>0$ is the bandwidth available to each pair through the static (electrical)
fabric, used in parallel with the circuits.  With the incidence matrix
$A\in\{0,1\}^{V\times\Eocs}$, the \emph{degree polytope} is
\begin{equation}
  \calW=\bigl\{\bfw\in\Simp:\;A\bfw\le (r/\mu)\,\one\bigr\}.
  \label{eq:calW}
\end{equation}
\end{definition}

$\calW$ is a closed $m$-convex polytope (not, in general, $m$-flat).  Integrality of the
circuit counts is relaxed here and restored in \cref{sec:rounding}.

\begin{assumption}[Strict feasibility]
\label{as:feasible}
$\calW$ contains a point with all coordinates positive.
\end{assumption}

For $\Eocs$ the complete graph, the uniform wiring has degree $2\mu/N\le r$, so
\cref{as:feasible} holds whenever $\mu\le Nr/2$.

\subsection{The Symmetric Prior}
\label{sec:prior}

The reference wiring is a positive $\bfq\in\Simp^\circ$.  Our default is the uniform
distribution on $\Eocs$: at pod granularity this is the logical capacity a non-blocking
Fat-Tree offers to uniform traffic, namely an equal share for every pod pair, and it is
invariant under all permutations of pods.  We write $\qW$ for its KL projection
\eqref{eq:iproj} onto $\calW$; when $\bfq\in\calW$ (as for the uniform prior under
\cref{as:feasible} with complete $\Eocs$), $\qW=\bfq$.

\begin{remark}[Priors with restricted support]
\label{rem:support}
If one wishes to use a prior supported on a subset $\Eocs_0\subsetneq\Eocs$ (for instance
the links of a physical Fat-Tree), then $\KL{\bfw}{\bfq}=+\infty$ for every wiring that puts
mass outside $\Eocs_0$, so the penalty forbids any departure from the support.  The
smoothed prior $\bfq_\epsilon=(1-\epsilon)\,\mathrm{Unif}(\Eocs_0)+\epsilon\,
\mathrm{Unif}(\Eocs\setminus\Eocs_0)$, $\epsilon\in(0,1)$, is positive and all results below
apply to it; the bounds depend on $\epsilon$ through $\KL{\cdot}{\bfq_\epsilon}\le\log(|\Eocs|/\epsilon)$.
\end{remark}

\subsection{Workload, DAG, and Smoothed Critical Path}
\label{sec:dag}

One training iteration is described by a DAG $\calG=(\mathcal{N},\calA)$ whose nodes are
either \emph{compute} nodes $v$ with fixed duration $c_v\ge0$ or \emph{communication} nodes
executing a phase $s(v)\in\calS$.  Phase $s$ moves volume $d_{s,e}\ge0$ over pair $e$;
$\Eocs_s=\{e:d_{s,e}>0\}$ and $n_s=|\Eocs_s|$.  A phase lasts as long as its slowest pair:
\begin{equation}
  T_s(\bfw)=\max_{e\in\Eocs_s}\frac{d_{s,e}}{\kappa_e(\bfw)}.
\end{equation}
Let $\calP$ be the set of paths from a source (node without predecessors) to a sink of
$\calG$, $H=\max_{P\in\calP}|\{v\in P:\ v\text{ communication}\}|$, and
$\ell_P=\sum_{v\in P}\mathrm{dur}_v$.  The iteration time is the critical-path length
$J(\bfw)=\max_{P\in\calP}\ell_P(\bfw)$.

\begin{assumption}[Modelling assumptions]
\label{as:model}
(i) Traffic of a phase uses the direct circuits of each pair together with the static
fabric (single-hop routing, capacities add as in \eqref{eq:capacity}); (ii) phases that may
run concurrently in $\calG$ do not contend for the same pair, or such contention is
neglected; (iii) each node starts as soon as all its predecessors have finished.
\end{assumption}

With $\lse_\tau(z_1,\dots,z_n)=\tau\log\sum_ie^{z_i/\tau}$, which satisfies
$\max_iz_i\le\lse_\tau(z)\le\max_iz_i+\tau\log n$, we smooth both maxima:
\begin{equation}
  T_s^{\tau_s}(\bfw)=\lse_{\tau_s}\Bigl(\tfrac{d_{s,e}}{\kappa_e(\bfw)}\Bigr)_{e\in\Eocs_s},
  \qquad
  U(\bfw)=J_\tau(\bfw)=\lse_\tau\bigl(\ell_P(\bfw)\bigr)_{P\in\calP},
  \label{eq:U}
\end{equation}
where communication durations in $\ell_P$ are $T^{\tau_s}_{s(v)}$.  Although $|\calP|$ may be
exponential, $U$ and its gradient are computed in $O(|\mathcal N|+|\calA|+\sum_sn_s)$ time
(\cref{prop:criticality}).

\begin{lemma}[Convexity, smoothness, approximation]
\label{lem:smooth}
Under \cref{as:model}, $U$ is convex and nonincreasing in every coordinate on
$\R^{\Eocs}_{\ge0}$.  With $\bar d_s=\max_ed_{s,e}$, $G_s=B\bar d_s/b_0^2$ and
$M_s=2B^2\bar d_s/b_0^3+G_s^2/\tau_s$, set $G_H=\max_{P}\sum_{v\in P}G_{s(v)}$ and
$M_H=\max_P\sum_{v\in P}M_{s(v)}$, the sums running over the communication nodes of $P$.  Then for all $\bfw\ge0$ and $\bfv\in\R^{\Eocs}$,
\begin{equation}
  \norm{\nabla U(\bfw)}_\infty\le G_H,\qquad
  \bfv^\top\nabla^2U(\bfw)\,\bfv\le L\norm{\bfv}_1^2,\qquad
  L=M_H+\frac{G_H^2}{\tau}.
  \label{eq:Lconst}
\end{equation}
Consequently $U(\bfy)\le U(\bfx)+\inner{\nabla U(\bfx)}{\bfy-\bfx}+L\,\KL{\bfy}{\bfx}$ for
all $\bfx,\bfy\in\Simp$ (relative smoothness), and
$0\le U(\bfw)-J(\bfw)\le\tau\log|\calP|+H\tau_s\log\max_sn_s$.
\end{lemma}

The proof is in \cref{app:proofs}.  Note that $G_H\le HG$ and $M_H\le HM$ for
$G=\max_sG_s$, $M=\max_sM_s$; the bound is conservative and in practice we use
backtracking (\cref{sec:algorithm}).

\begin{proposition}[Soft criticality and forward--backward gradient]
\label{prop:criticality}
Let $\sigma_\tau(P)=\exp((\ell_P-U)/\tau)$, a Gibbs distribution on $\calP$, and
$\gamma_v=\sigma_\tau(\{P\ni v\})$ the \emph{soft criticality} of node $v$.  Then
$\partial U/\partial\mathrm{dur}_v=\gamma_v$ and
\begin{equation}
  \nabla U(\bfw)=\sum_{v\ \mathrm{comm.}}\gamma_v\,\nabla T^{\tau_s}_{s(v)}(\bfw),
  \qquad
  \frac{\partial T^{\tau_s}_s}{\partial w_e}
  =-\,\pi_{s,e}\,\frac{d_{s,e}B}{\kappa_e(\bfw)^2},
  \label{eq:gradU}
\end{equation}
where $\pi_{s,\cdot}$ is the softmax over $\Eocs_s$ of $d_{s,e}/(\tau_s\kappa_e)$.  The
quantities $a_v=\mathrm{dur}_v+\lse_\tau(a_u)_{u\in\mathrm{pred}(v)}$ (forward, $a_v=\mathrm{dur}_v$
at sources) satisfy $U=\lse_\tau(a_t)_{t\ \mathrm{sink}}$, and $\gamma$ obeys the backward
recursion $\gamma_t=e^{(a_t-U)/\tau}$ for sinks plus
$\gamma_u=\sum_{v\in\mathrm{succ}(u)}\gamma_v\,e^{a_u/\tau}/\sum_{u'\in\mathrm{pred}(v)}e^{a_{u'}/\tau}$.
If the critical path is unique, $\gamma_v\to1$ on it and $\gamma_v\to0$ elsewhere as $\tau\to0$.
\end{proposition}

\begin{proof}
$U=\tau\log\sum_Pe^{\ell_P/\tau}$ gives $\partial U/\partial\ell_P=\sigma_\tau(P)$ and
$\partial\ell_P/\partial\mathrm{dur}_v=\mathbf 1[v\in P]$, hence $\partial U/\partial
\mathrm{dur}_v=\gamma_v$; \eqref{eq:gradU} follows from the chain rule and the derivative of
$\lse$.  By induction over a topological order,
$e^{a_v/\tau}=\sum_{\text{paths }Q\text{ from a source to }v}e^{\ell_Q/\tau}$, because
$e^{a_v/\tau}=e^{\mathrm{dur}_v/\tau}\sum_{u\in\mathrm{pred}(v)}e^{a_u/\tau}$; summing over sinks
gives $U$.  The backward recursion is reverse-mode differentiation of the forward
recursion, and $\gamma_u$ equals the $\sigma_\tau$-probability that the path passes $u$
because each factor $e^{a_u/\tau}/\sum_{u'}e^{a_{u'}/\tau}$ is the conditional probability that
the path enters $v$ from $u$.
\end{proof}

\begin{remark}[The DAG as an exponential family]
Phase durations are consequences of the wiring rather than free decision variables, so
precedence constraints are not imposed on a separate schedule; they enter the problem
exactly through the critical-path functional $J$.  The geometric object that the DAG induces
is the exponential family $\{\sigma_\tau\}$ on $\calP$, with sufficient statistic the path
length and natural parameter $1/\tau$; the criticalities $\gamma_v$ are its marginals, and the
forward--backward recursion is the familiar sum--product computation of these marginals.
\end{remark}

\section{Free Energy and Temperature}
\label{sec:freeenergy}

\begin{definition}[OCS free energy]
\label{def:free_energy}
For an inverse temperature $\beta_1\in(0,\infty)$ and $\lambda=1/\beta_1$, let
\begin{equation}
  \boxed{\;F_{\beta_1}(\bfw)=U(\bfw)+\frac{1}{\beta_1}\KL{\bfw}{\bfq},\qquad
  \bfw^\star_{\beta_1}=\argmin_{\bfw\in\calW}F_{\beta_1}(\bfw).\;}
  \label{eq:free_energy}
\end{equation}
\end{definition}

As in statistical mechanics, $\beta_1^{-1}$ multiplies the entropic (KL) term, so
$\beta_1\to0$ (high temperature) pins the wiring to the prior and $\beta_1\to\infty$ (low
temperature) yields pure cost minimisation (\cref{tab:temperatures}).

\begin{table}[h]
\centering
\caption{Role of the inverse temperature $\beta_1$ (cf.\ \cref{thm:limits}).}
\label{tab:temperatures}
\renewcommand{\arraystretch}{1.25}
\begin{tabular}{lll}
\toprule
& $\beta_1\to0$ ($\lambda\to\infty$) & $\beta_1\to\infty$ ($\lambda\to0$)\\
\midrule
Optimal wiring & $\qW$ (static symmetric prior) & min-KL minimiser of $U$ on $\calW$\\
Cost $U(\bfw^\star_{\beta_1})$ & $U(\qW)$ & $\min_\calW U$\\
Linear rate of \cref{thm:convergence} & fast, $(1+\eta\lambda)^{-1}$ small & slow, tends to sublinear\\
\bottomrule
\end{tabular}
\end{table}

\begin{theorem}[Well-posedness]
\label{thm:wellposed}
Under \cref{as:feasible,as:model}, for every $\beta_1>0$ the minimiser
$\bfw^\star=\bfw^\star_{\beta_1}$ exists, is unique and has all coordinates positive, and for
all $\bfx\in\calW$
\begin{equation}
  F_{\beta_1}(\bfx)\;\ge\;F_{\beta_1}(\bfw^\star)+\lambda\,\KL{\bfx}{\bfw^\star}.
  \label{eq:rel_strong}
\end{equation}
\end{theorem}

\begin{proof}
$F_{\beta_1}$ is continuous on the compact set $\calW$, so a minimiser exists; it is positive
by \cref{lem:positivity} with $c=\lambda$ and $\psi=U-\lambda\inner{\log\bfq}{\cdot}$.  By
convexity, $U(\bfx)\ge U(\bfw^\star)+\inner{\nabla U(\bfw^\star)}{\bfx-\bfw^\star}$, and by
\eqref{eq:three_point}, $\KL{\bfx}{\bfq}=\KL{\bfw^\star}{\bfq}+\inner{\log(\bfw^\star/\bfq)}
{\bfx-\bfw^\star}+\KL{\bfx}{\bfw^\star}$.  Adding,
$F_{\beta_1}(\bfx)\ge F_{\beta_1}(\bfw^\star)+\inner{\nabla F_{\beta_1}(\bfw^\star)}{\bfx-\bfw^\star}
+\lambda\KL{\bfx}{\bfw^\star}$ (the constant in $\nabla h=\log\bfw+\one$ is orthogonal to
$\bfx-\bfw^\star$).  The first-order optimality condition at the positive point
$\bfw^\star$ makes the inner product nonnegative, giving \eqref{eq:rel_strong}; uniqueness
follows since $\KL{\bfx}{\bfw^\star}>0$ for $\bfx\ne\bfw^\star$.
\end{proof}

\begin{theorem}[Temperature: monotonicity and limits]
\label{thm:limits}
Let $U_{\min}=\min_\calW U$ and let $\bfw^\circ$ be the unique minimiser of
$\KL{\cdot}{\bfq}$ over the convex set $\argmin_\calW U$.
\begin{enumerate}[label=(\roman*)]
\item $\beta_1\mapsto U(\bfw^\star_{\beta_1})$ is nonincreasing and
  $\beta_1\mapsto\KL{\bfw^\star_{\beta_1}}{\bfq}$ is nondecreasing.
\item (High temperature) $\KL{\bfw^\star_{\beta_1}}{\qW}\le\beta_1\bigl(U(\qW)-U_{\min}\bigr)$;
  hence $\bfw^\star_{\beta_1}\to\qW$ as $\beta_1\to0$.
\item (Low temperature) $0\le U(\bfw^\star_{\beta_1})-U_{\min}\le\beta_1^{-1}\,
  \KL{\bfw^\circ}{\bfq}$ and $\KL{\bfw^\star_{\beta_1}}{\bfq}\le\KL{\bfw^\circ}{\bfq}$; hence
  $\bfw^\star_{\beta_1}\to\bfw^\circ$ as $\beta_1\to\infty$.
\end{enumerate}
\end{theorem}

\begin{proof}
Write $K(\bfw)=\KL{\bfw}{\bfq}$ and, for $\lambda>\lambda'$, $\bfw=\bfw^\star_{1/\lambda}$,
$\bfw'=\bfw^\star_{1/\lambda'}$.

(i) Optimality gives $U(\bfw)+\lambda K(\bfw)\le U(\bfw')+\lambda K(\bfw')$ and
$U(\bfw')+\lambda'K(\bfw')\le U(\bfw)+\lambda'K(\bfw)$.  Adding them,
$(\lambda-\lambda')(K(\bfw)-K(\bfw'))\le0$, so $K(\bfw)\le K(\bfw')$; the second inequality
then gives $U(\bfw')-U(\bfw)\le\lambda'(K(\bfw)-K(\bfw'))\le0$.

(ii) Comparing with $\qW\in\calW$: $U(\bfw)+\lambda K(\bfw)\le U(\qW)+\lambda K(\qW)$.  By
\cref{thm:pythagoras}, $K(\bfw)\ge\KL{\bfw}{\qW}+K(\qW)$.  Hence
$\lambda\KL{\bfw}{\qW}\le U(\qW)-U(\bfw)\le U(\qW)-U_{\min}$.  Convergence follows from
Pinsker's inequality.

(iii) Comparing with any $\bfw^m\in\argmin_\calW U$:
$U(\bfw)+\lambda K(\bfw)\le U_{\min}+\lambda K(\bfw^m)\le U(\bfw)+\lambda K(\bfw^m)$, which
gives both $U(\bfw)-U_{\min}\le\lambda K(\bfw^m)$ and $K(\bfw)\le K(\bfw^m)$; take
$\bfw^m=\bfw^\circ$.  Let $\bar\bfw$ be a limit point as $\lambda\to0$.  By continuity,
$U(\bar\bfw)=U_{\min}$ and $K(\bar\bfw)\le K(\bfw^\circ)$; since $\bfw^\circ$ is the unique
KL-minimal element of the convex set $\argmin_\calW U$ ($K$ is strictly convex),
$\bar\bfw=\bfw^\circ$, and compactness gives convergence of the whole family.
\end{proof}

\begin{theorem}[Gain of adaptation]
\label{thm:jeffreys}
For every $\beta_1>0$, with $\bfw^\star=\bfw^\star_{\beta_1}$,
\begin{equation}
  U(\qW)-U(\bfw^\star)\;\ge\;\frac{1}{\beta_1}\Bigl(\KL{\bfw^\star}{\qW}+\KL{\qW}{\bfw^\star}\Bigr).
  \label{eq:jeffreys}
\end{equation}
That is, the reduction of the iteration cost obtained by departing from the static prior is
at least $\beta_1^{-1}$ times the Jeffreys (symmetrised KL) divergence between the static and
the adaptive wiring.
\end{theorem}

\begin{proof}
Apply \eqref{eq:rel_strong} with $\bfx=\qW$:
$U(\qW)+\lambda K(\qW)\ge U(\bfw^\star)+\lambda K(\bfw^\star)+\lambda\KL{\qW}{\bfw^\star}$.
Rearranging and using $K(\bfw^\star)-K(\qW)\ge\KL{\bfw^\star}{\qW}$ (\cref{thm:pythagoras})
gives \eqref{eq:jeffreys}.
\end{proof}

\begin{proposition}[Symmetry inheritance]
\label{prop:symmetry}
Let $\mathfrak G$ be a group of permutations of $\Eocs$ acting by
$(g\bfw)_e=w_{g^{-1}e}$.  If $U(g\bfw)=U(\bfw)$, $g\calW=\calW$ and $g\bfq=\bfq$ for all
$g\in\mathfrak G$, then $g\bfw^\star_{\beta_1}=\bfw^\star_{\beta_1}$ for all $g$ and
$\beta_1$; i.e.\ the optimal wiring is constant on the $\mathfrak G$-orbits of $\Eocs$.
\end{proposition}

\begin{proof}
$\KL{g\bfw}{g\bfq}=\KL{\bfw}{\bfq}$, so $F_{\beta_1}(g\bfw^\star)=F_{\beta_1}(\bfw^\star)$ with
$g\bfw^\star\in\calW$; uniqueness (\cref{thm:wellposed}) gives $g\bfw^\star=\bfw^\star$.
\end{proof}

\begin{corollary}[Uniform traffic]
If $\Eocs$ is complete, $\bfq$ is uniform, and the workload is invariant under all
permutations of pods (e.g.\ uniform all-to-all phases), then $\bfw^\star_{\beta_1}=\bfq$ for every
$\beta_1$: the static symmetric wiring is optimal.
\end{corollary}

\begin{proof}
The symmetric group on pods acts transitively on pairs, so by \cref{prop:symmetry}
$\bfw^\star$ is constant on $\Eocs$, i.e.\ uniform.
\end{proof}

\begin{proposition}[When is the static prior optimal?]
\label{prop:staticopt}
If $\qW$ minimises $U$ over $\calW$, then $\bfw^\star_{\beta_1}=\qW$ for all $\beta_1>0$.  If
moreover $\bfq\in\calW$ (so $\qW=\bfq$), the converse holds: $\bfw^\star_{\beta_1}=\bfq$ for some
$\beta_1>0$ if and only if $\bfq\in\argmin_\calW U$.
\end{proposition}

\begin{proof}
The optimality condition of $\bfw\in\calW$ for $F_{\beta_1}$ is
$\inner{\nabla U(\bfw)+\lambda\log(\bfw/\bfq)}{\bfx-\bfw}\ge0$ for all $\bfx\in\calW$.  At
$\bfw=\qW$ the KL-projection condition gives $\inner{\log(\qW/\bfq)}{\bfx-\qW}\ge0$; if also
$\inner{\nabla U(\qW)}{\bfx-\qW}\ge0$, the sum is nonnegative.  If $\bfq\in\calW$ the log term
vanishes at $\bfw=\bfq$, so the condition for $F_{\beta_1}$ coincides with the optimality
condition for $U$.
\end{proof}

\begin{remark}
Taken together, \cref{prop:symmetry,prop:staticopt,thm:jeffreys} delimit the role of the
symmetric reference precisely: the static wiring is optimal exactly when it already minimises
the cost, which symmetric workloads guarantee, and otherwise the loss incurred by retaining it
is bounded below by a divergence between the two wirings.  \Cref{fig:temperature} illustrates
\cref{thm:limits}.
\end{remark}

\begin{remark}[Scope: the continuous relaxation]
\Cref{thm:wellposed,thm:limits,thm:jeffreys} concern the continuous relaxation, in which
$\mu w_e$ need not be an integer.  Deployable wirings consist of integer numbers of circuits,
and the properties established here are not inherited automatically by integer optima; the
relation between the two problems is quantified in \cref{prop:rounding} and examined
empirically in \cref{sec:exp_integer}.
\end{remark}

\section{Entropic Proximal-Gradient Algorithm and Convergence}
\label{sec:algorithm}

\subsection{The Update as a KL Projection}

Given $\bfw^k\in\calW$ with positive coordinates, $\bfg^k=\nabla U(\bfw^k)$ and a step
$\eta_k>0$, the Bregman proximal-gradient step \cite{bauschke2017descent,lu2018relatively}
with the entropy as reference function is
\begin{equation}
  \bfw^{k+1}=\argmin_{\bfw\in\calW}\Bigl\{\inner{\bfg^k}{\bfw}+\lambda\KL{\bfw}{\bfq}
  +\tfrac{1}{\eta_k}\KL{\bfw}{\bfw^k}\Bigr\}.
  \label{eq:prox_step}
\end{equation}

\begin{lemma}[Closed form]
\label{lem:closedform}
Let $a_k=1/(1+\eta_k\lambda)$ and $b_k=\eta_k/(1+\eta_k\lambda)$, and define the geometric
mixture
\begin{equation}
  \hat w^k_e=\frac{(w^k_e)^{a_k}\,q_e^{\,1-a_k}\,e^{-b_kg^k_e}}
  {\sum_{f}(w^k_f)^{a_k}\,q_f^{\,1-a_k}\,e^{-b_kg^k_f}}.
  \label{eq:geomix}
\end{equation}
Then $\bfw^{k+1}=\argmin_{\bfw\in\calW}\KL{\bfw}{\hat\bfw^k}$, the KL projection of $\hat\bfw^k$
onto $\calW$.  In particular, without degree constraints the update is the multiplicative
rule $\bfw^{k+1}=\hat\bfw^k$.
\end{lemma}

\begin{proof}
With $c=\lambda+1/\eta_k$, the objective of \eqref{eq:prox_step} equals, up to a constant,
\[
  c\sum_ew_e\log w_e-\sum_ew_e\bigl(\lambda\log q_e+\eta_k^{-1}\log w^k_e-g^k_e\bigr).
\]  Dividing by $c$, the bracket becomes $\log\hat w^k_e+\log Z$ with $Z$ the
normaliser in \eqref{eq:geomix}; since $\sum_ew_e=1$, the objective is
$c\,\KL{\bfw}{\hat\bfw^k}-c\log Z$.
\end{proof}

\begin{lemma}[Dual of the KL projection]
\label{lem:dual}
Let $\hat\bfw>0$ and $\calW$ as in \eqref{eq:calW} with $\bfb=(r/\mu)\one$, satisfying
\cref{as:feasible}.  For $\bfy\in\R^V_{\ge0}$ let
$w_e(\bfy)\propto\hat w_e\exp(-(A^\top\bfy)_e)$ and
\begin{equation}
  \phi(\bfy)=-\log\sum_e\hat w_e\,e^{-(A^\top\bfy)_e}-\inner{\bfb}{\bfy}.
\end{equation}
Then $\phi$ is concave and smooth with $\nabla\phi(\bfy)=A\bfw(\bfy)-\bfb$, the problem
$\max_{\bfy\ge0}\phi(\bfy)$ attains its maximum at some $\bfy^\star$, strong duality holds,
and $\bfw(\bfy^\star)=\argmin_{\calW}\KL{\cdot}{\hat\bfw}$, with complementary slackness
$y^\star_v\,(A\bfw(\bfy^\star)-\bfb)_v=0$.
\end{lemma}

\begin{proof}
For fixed $\bfy\ge0$, minimising the Lagrangian $\KL{\bfw}{\hat\bfw}+\inner{\bfy}{A\bfw-\bfb}$ over
$\Simp$ is a Gibbs variational problem with unique solution $\bfw(\bfy)$ and value
$\phi(\bfy)$; $\phi$ is a pointwise infimum of affine functions of $\bfy$, hence concave, and
its gradient follows by Danskin's theorem.  The primal is convex with affine inequality
constraints and, by \cref{as:feasible}, has a feasible point in the relative interior of
the domain $\Simp$ of $\KL{\cdot}{\hat\bfw}$; the refined Slater condition
\cite[\S5.2.3]{boyd2004convex} gives strong duality and dual attainment.  Since the
Lagrangian minimiser is unique, $\bfw(\bfy^\star)$ is the primal solution.
\end{proof}

The dual has only $|V|$ variables and is solved by a quasi-Newton method with bound
constraints; each evaluation costs $O(|\Eocs|)$.

\begin{algorithm}[t]
\caption{Entropic proximal gradient for $\min_{\calW}U+\beta_1^{-1}\KL{\cdot}{\bfq}$}
\label{alg:main}
\begin{algorithmic}[1]
\State \textbf{Input:} prior $\bfq>0$, $\lambda=1/\beta_1\ge0$, initial step $\eta_0>0$
\State $\bfw^0\gets\qW$ \Comment{static symmetric wiring; $=\bfq$ if $\bfq\in\calW$}
\For{$k=0,1,2,\dots$}
  \State $\bfg^k\gets\nabla U(\bfw^k)$ by the forward--backward pass (\cref{prop:criticality})
  \State $\eta\gets2\eta_{k-1}$ \Comment{$\eta_{-1}=\eta_0/2$}
  \Repeat
    \State form $\hat\bfw$ by \eqref{eq:geomix} with step $\eta$; \;
           $\bfw^+\gets\argmin_{\calW}\KL{\cdot}{\hat\bfw}$ via \cref{lem:dual}
    \State \textbf{if} $U(\bfw^+)>U(\bfw^k)+\inner{\bfg^k}{\bfw^+-\bfw^k}+\eta^{-1}\KL{\bfw^+}{\bfw^k}$
           \textbf{then} $\eta\gets\eta/2$
  \Until{the test in the previous line passes}
  \State $\eta_k\gets\eta$, \; $\bfw^{k+1}\gets\bfw^+$
\EndFor
\end{algorithmic}
\end{algorithm}

\subsection{Convergence}

\begin{lemma}[Three-point property]
\label{lem:threepoint}
For the step \eqref{eq:prox_step}, $\bfw^{k+1}$ is positive and for every $\bfx\in\calW$,
\begin{equation}
\begin{aligned}
  \inner{\bfg^k}{\bfx}+\lambda\KL{\bfx}{\bfq}+\tfrac1{\eta_k}\KL{\bfx}{\bfw^k}
  \;\ge{}&
  \inner{\bfg^k}{\bfw^{k+1}}+\lambda\KL{\bfw^{k+1}}{\bfq}+\tfrac1{\eta_k}\KL{\bfw^{k+1}}{\bfw^k}\\
  &+\Bigl(\lambda+\tfrac1{\eta_k}\Bigr)\KL{\bfx}{\bfw^{k+1}}.
\end{aligned}
  \label{eq:threepoint}
\end{equation}
\end{lemma}

\begin{proof}
The objective $\varphi$ of \eqref{eq:prox_step} has the form $c\,h+\psi$ with
$c=\lambda+\eta_k^{-1}$ and $\psi$ linear, so its minimiser is positive by
\cref{lem:positivity}.  For linear $\psi$ the Taylor expansion of $\varphi$ around the
positive point $\bfw^{k+1}$ is exact:
$\varphi(\bfx)=\varphi(\bfw^{k+1})+\inner{\nabla\varphi(\bfw^{k+1})}{\bfx-\bfw^{k+1}}
+c\bigl(h(\bfx)-h(\bfw^{k+1})-\inner{\nabla h(\bfw^{k+1})}{\bfx-\bfw^{k+1}}\bigr)$, and the last
bracket equals $\KL{\bfx}{\bfw^{k+1}}$ on $\Simp$.  The first-order optimality condition makes
the inner product nonnegative.
\end{proof}

\begin{theorem}[Monotone descent and linear convergence]
\label{thm:convergence}
Let \cref{as:feasible,as:model} hold, let $\bfw^\star$ be the minimiser of $F=F_{\beta_1}$ on
$\calW$ and $F^\star=F(\bfw^\star)$.  Suppose every step satisfies the acceptance test of
\cref{alg:main},
\begin{equation}
  U(\bfw^{k+1})\le U(\bfw^k)+\inner{\bfg^k}{\bfw^{k+1}-\bfw^k}+\eta_k^{-1}\KL{\bfw^{k+1}}{\bfw^k},
  \label{eq:accept}
\end{equation}
which holds in particular whenever $\eta_k\le1/L$ (\cref{lem:smooth}); backtracking therefore
terminates with $\eta_k\ge\eta_{\min}:=\min\{\eta_0,1/(2L)\}$.  Then for all $k$:
\begin{enumerate}[label=(\alph*)]
\item $F(\bfw^{k+1})\le F(\bfw^k)-(\lambda+\eta_k^{-1})\KL{\bfw^k}{\bfw^{k+1}}$;
\item $F(\bfw^{k+1})-F^\star\le\eta_k^{-1}\KL{\bfw^\star}{\bfw^k}-(\lambda+\eta_k^{-1})
  \KL{\bfw^\star}{\bfw^{k+1}}$, hence
  $\KL{\bfw^\star}{\bfw^{k+1}}\le(1+\eta_k\lambda)^{-1}\KL{\bfw^\star}{\bfw^k}$;
\item consequently
  $\KL{\bfw^\star}{\bfw^k}\le\prod_{i<k}(1+\eta_i\lambda)^{-1}\KL{\bfw^\star}{\bfw^0}
  \le(1+\eta_{\min}\lambda)^{-k}\KL{\bfw^\star}{\bfw^0}$ and
  $F(\bfw^{k+1})-F^\star\le\eta_{\min}^{-1}(1+\eta_{\min}\lambda)^{-k}\KL{\bfw^\star}{\bfw^0}$;
\item if $\lambda=0$ and $\eta_k\equiv\eta\le1/L$, then
  $F(\bfw^K)-F^\star\le\KL{\bfw^\star}{\bfw^0}/(K\eta)$.
\end{enumerate}
With $\bfw^0=\bfq$ uniform, $\KL{\bfw^\star}{\bfw^0}\le\log|\Eocs|$.
\end{theorem}

\begin{proof}
Fix $\bfx\in\calW$.  Convexity of $U$ gives
$U(\bfx)\ge U(\bfw^k)+\inner{\bfg^k}{\bfx-\bfw^k}$.  Combining \eqref{eq:accept}, the three-point
property \eqref{eq:threepoint} and this inequality,
\begin{align*}
F(\bfw^{k+1})
&\le U(\bfw^k)+\inner{\bfg^k}{\bfw^{k+1}-\bfw^k}+\eta_k^{-1}\KL{\bfw^{k+1}}{\bfw^k}
   +\lambda\KL{\bfw^{k+1}}{\bfq}\\
&\le U(\bfw^k)+\inner{\bfg^k}{\bfx-\bfw^k}+\lambda\KL{\bfx}{\bfq}+\eta_k^{-1}\KL{\bfx}{\bfw^k}
   -(\lambda+\eta_k^{-1})\KL{\bfx}{\bfw^{k+1}}\\
&\le F(\bfx)+\eta_k^{-1}\KL{\bfx}{\bfw^k}-(\lambda+\eta_k^{-1})\KL{\bfx}{\bfw^{k+1}}.
\end{align*}
Taking $\bfx=\bfw^k$ gives (a).  Taking $\bfx=\bfw^\star$ gives the first part of (b); since
$F(\bfw^{k+1})\ge F^\star$, the left side is nonnegative, which gives the contraction.
Iterating gives the first part of (c); the bound on $F$ follows from (b) and
$\eta_k\ge\eta_{\min}$.  For (d), summing (b) with $\lambda=0$ over $k=0,\dots,K-1$ telescopes
to $\sum_{k=1}^{K}(F(\bfw^k)-F^\star)\le\KL{\bfw^\star}{\bfw^0}/\eta$, and by (a) the left side is
at least $K(F(\bfw^K)-F^\star)$.  If $\eta_k\le1/L$, \eqref{eq:accept} follows from
\cref{lem:smooth}; halving from any $\eta>1/L$ stops at some $\eta\in(1/(2L),1/L]$.  Finally,
$\KL{\bfw^\star}{\bfq}=\sum_ew^\star_e\log(|\Eocs|w^\star_e)\le\log|\Eocs|$ for uniform $\bfq$.
\end{proof}

In practice the KL projection, and hence the step \eqref{eq:prox_step}, is computed only
approximately.  The following result quantifies the effect.

\begin{theorem}[Inexact steps]
\label{thm:inexact}
Suppose that, instead of \eqref{eq:threepoint}, the computed iterate $\bfw^{k+1}\in\calW$
satisfies \eqref{eq:accept} and the three-point inequality \eqref{eq:threepoint} with
$\bfx=\bfw^\star$ up to an additive error $\delta_k\ge0$ on its right-hand side.  Then
\begin{equation}
  \KL{\bfw^\star}{\bfw^{k+1}}\le\frac{\KL{\bfw^\star}{\bfw^k}+\eta_k\delta_k}{1+\eta_k\lambda}.
  \label{eq:inexact}
\end{equation}
In particular, if $\lambda>0$, $\eta_k\ge\eta_{\min}$ and $\delta_k\le\delta$, then
$\KL{\bfw^\star}{\bfw^k}\le(1+\eta_{\min}\lambda)^{-k}\KL{\bfw^\star}{\bfw^0}+\delta/\lambda$
whenever $\eta_k\equiv\eta_{\min}$, and in general
$\limsup_k\KL{\bfw^\star}{\bfw^k}\le\delta\,\sup_k\eta_k/(\eta_{\min}\lambda)$.  If
$\sum_k\eta_k\delta_k<\infty$, then $\KL{\bfw^\star}{\bfw^k}\to0$.
\end{theorem}

\begin{proof}
Repeating the chain of inequalities in the proof of \cref{thm:convergence} with
$\bfx=\bfw^\star$ and the perturbed three-point inequality gives
$F(\bfw^{k+1})-F^\star\le\eta_k^{-1}\KL{\bfw^\star}{\bfw^k}+\delta_k-(\lambda+\eta_k^{-1})
\KL{\bfw^\star}{\bfw^{k+1}}$.  Since the left-hand side is nonnegative, multiplying by $\eta_k$
yields \eqref{eq:inexact}.  Writing $a_k=\KL{\bfw^\star}{\bfw^k}$ and $\rho_k=(1+\eta_k\lambda)^{-1}$,
we have $a_{k+1}\le\rho_ka_k+\rho_k\eta_k\delta_k$.  For constant $\eta_k=\eta_{\min}$ and
$\rho=(1+\eta_{\min}\lambda)^{-1}$, unrolling gives
$a_k\le\rho^ka_0+\eta_{\min}\delta\sum_{j\ge1}\rho^j=\rho^ka_0+\eta_{\min}\delta\,\rho/(1-\rho)
=\rho^ka_0+\delta/\lambda$.  In general, put $\bar\rho=(1+\eta_{\min}\lambda)^{-1}\ge\rho_k$ and
$\bar\eta=\sup_k\eta_k$.  Since $\eta\mapsto\eta/(1+\eta\lambda)$ is increasing,
$\rho_k\eta_k\delta_k\le\bar\eta\delta/(1+\bar\eta\lambda)$, whence
$\limsup_ka_k\le\frac{\bar\eta\delta}{(1+\bar\eta\lambda)(1-\bar\rho)}
=\frac{\bar\eta\delta(1+\eta_{\min}\lambda)}{(1+\bar\eta\lambda)\,\eta_{\min}\lambda}
\le\frac{\bar\eta\,\delta}{\eta_{\min}\lambda}$.  Finally, if $\sum_k\eta_k\delta_k<\infty$, the
recursion $a_{k+1}\le\bar\rho a_k+\eta_k\delta_k$ with $\bar\rho<1$ and summable forcing gives
$a_k\to0$.
\end{proof}

\Cref{thm:inexact} is a conditional error-propagation result: it asserts what follows
\emph{if} the computed iterates satisfy the perturbed three-point inequality with defects
$\delta_k$.  We do not derive $\delta_k$ from the stopping rule of the dual solver.  Whether
the implementation meets the hypothesis is instead verified numerically: for the runs of
\cref{fig:convergence} we evaluated the defect of the three-point inequality at a reference
optimum after every step, and obtained $\max_k\eta_k\delta_k\le1.6\times10^{-10}$ and
$\sum_k\eta_k\delta_k\le1.0\times10^{-9}$ over 400 iterations, with degree-constraint residuals
below $2.4\times10^{-8}$.  By \eqref{eq:inexact} the resulting error floor lies several orders of
magnitude below the accuracies reported in \cref{sec:experiments}.

\begin{remark}[Complexity and scope]
Each iteration costs one forward--backward pass ($O(|\mathcal N|+|\calA|+\sum_sn_s)$), one
mixture \eqref{eq:geomix} ($O(|\Eocs|)$) and a $|V|$-dimensional concave maximisation; by (c)
an $\varepsilon$-accurate solution needs $O\bigl((1+L\beta_1)\log(\log|\Eocs|/(\eta_{\min}\varepsilon))\bigr)$
iterations.  These guarantees concern the continuous relaxation; the integer topology problem
is treated next and compared with an exact MILP in \cref{sec:exp_integer}.  Theorem
\ref{thm:convergence} is an instance of the analysis of Bregman proximal gradient methods
under relative smoothness \cite{bauschke2017descent,lu2018relatively}; its role here is to
certify the algorithm actually used, not to contribute new convergence theory.
\end{remark}

\subsection{Integer Circuits}
\label{sec:rounding}

A deployable wiring has integer circuit counts $\bfx\in\mathbb Z^{\Eocs}_{\ge0}$ with
$A\bfx\le r\one$ and $\one^\top\bfx\le\mu$; its capacity is $b_0+cx_e$, i.e.\ that of
$\bfw=\bfx/\mu$.  Let $\calW'=\{\bfw\ge0:\one^\top\bfw\le1,\ A\bfw\le(r/\mu)\one\}\supseteq\calW$.

\begin{proposition}[Relaxation bound and floor rounding]
\label{prop:rounding}
Let $J^\star_{\mathbb Z}$ be the optimal iteration time over integer wirings and
$\bfw'\in\argmin_{\calW'}J$.  Then $J(\bfw')\le J^\star_{\mathbb Z}$, and the floor rounding
$\bfx=\lfloor\mu\bfw'\rfloor$ is feasible and satisfies
\begin{equation}
  J(\bfx/\mu)\;\le\;\Bigl(1+\frac{c}{b_0}\Bigr)J(\bfw')\;\le\;\Bigl(1+\frac{c}{b_0}\Bigr)J^\star_{\mathbb Z}.
\end{equation}
\end{proposition}

\begin{proof}
Every integer-feasible $\bfx$ gives $\bfx/\mu\in\calW'$, hence the first inequality.  Flooring
preserves the degree and budget constraints.  For each pair,
$b_0+c\lfloor\mu w'_e\rfloor\ge\max\{b_0,\kappa_e(\bfw')-c\}$, so
$\kappa_e(\bfw')/\kappa_e(\bfx/\mu)\le\max\{\kappa_e/(\kappa_e-c)\ \text{if}\ \kappa_e-c\ge b_0;\ \kappa_e/b_0\ \text{otherwise}\}
\le1+c/b_0$.  Each transfer time, hence each phase duration, grows by at most this factor;
compute durations are unchanged, so every path length and therefore $J$ grows by at most
$1+c/b_0$.
\end{proof}

When the pod bandwidth $rc$ is fixed and the number of ports $r$ grows, $c/b_0\to0$ and floor
rounding becomes asymptotically exact.  For small $r$ we improve the rounded solution by a
first-improvement local search over additions and degree-preserving 2-swaps
$\{(i,j),(k,l)\}\to\{(i,k),(j,l)\}$, accepting moves on a smoothed objective with an annealed
temperature (smoothing removes the plateaus of the hard maximum).  The local search carries
no guarantee; \cref{sec:exp_integer} quantifies its gap to the exact optimum.

\section{When to Reconfigure: Regret and Likelihood-Ratio Detection}
\label{sec:reconf}

\subsection{Setting}

Time is indexed by training iterations $t=1,\dots,T$.  The workload of iteration $t$ is
described by a distribution $\bfP_t$ over \emph{cells} $(s,e)$, $s\in\calS$, $e\in\Eocs$, with
$d_{s,e,t}=V\,P_t(s,e)$ for a known total volume $V$ (bytes per iteration are available from
counters; their \emph{distribution} must be estimated).  Let $M$ be the number of cells,
$U_t(\bfw)=U(\bfw;\bfP_t)$, $F_t=U_t+\lambda\KL{\cdot}{\bfq}$ and $F^\star_t=\min_\calW F_t$.
A policy keeps the deployed wiring $\bfw_t$ unchanged or reconfigures it at cost $C>0$
(downtime in seconds); a wiring computed at time $s$ is \emph{$\varepsilon_{\rm opt}$-optimal}
if $F_s(\bfw_s)\le F_s^\star+\varepsilon_{\rm opt}$.  The \emph{regret} against the
zero-switching-cost dynamic optimum is
\begin{equation}
  R_T=\sum_{t=1}^T\bigl(F_t(\bfw_t)-F_t^\star\bigr)+C\,N_T,
\end{equation}
where $N_T$ is the number of reconfigurations, and the workload path length is
\[
  V_1=\textstyle\sum_{t=2}^T\norm{\bfP_t-\bfP_{t-1}}_1.
\]

\begin{assumption}[Workload sensitivity]
\label{as:lipschitz}
$|U(\bfw;\bfP)-U(\bfw;\bfP')|\le\Gamma\norm{\bfP-\bfP'}_1$ for all $\bfw\in\calW$.
\end{assumption}

\begin{lemma}
\label{lem:gamma}
For the smoothed critical-path cost \eqref{eq:U}, \cref{as:lipschitz} holds with
$\Gamma=V/b_0$.
\end{lemma}

\begin{proof}
$\partial T_s^{\tau_s}/\partial d_{s,e}=\pi_{s,e}/\kappa_e\le\pi_{s,e}/b_0$ with $\sum_e\pi_{s,e}=1$,
so $|T_s^{\tau_s}(\bfd)-T_s^{\tau_s}(\bfd')|\le b_0^{-1}\max_e|d_{s,e}-d'_{s,e}|$.  $\lse_\tau$ is
$1$-Lipschitz in the sup-norm of its arguments, so
$|U(\bfd)-U(\bfd')|\le\max_P|\ell_P-\ell'_P|\le b_0^{-1}\sum_s\max_e|d_{s,e}-d'_{s,e}|
\le b_0^{-1}\norm{\bfd-\bfd'}_1=(V/b_0)\norm{\bfP-\bfP'}_1$.
\end{proof}

\begin{lemma}[Holding regret]
\label{lem:holding}
If $\bfw_s$ is $\varepsilon_{\rm opt}$-optimal for $F_s$, then for every $t$,
\begin{equation}
  F_t(\bfw_s)-F_t^\star\le2\Gamma\norm{\bfP_t-\bfP_s}_1+\varepsilon_{\rm opt}.
\end{equation}
\end{lemma}

\begin{proof}
Let $\bfw^\star_t$ minimise $F_t$.  The KL terms do not depend on $t$, so
$F_t(\bfw_s)-F_t^\star=[U_t(\bfw_s)-U_s(\bfw_s)]+[F_s(\bfw_s)-F_s(\bfw^\star_t)]+[U_s(\bfw^\star_t)-U_t(\bfw^\star_t)]$;
the outer brackets are at most $\Gamma\norm{\bfP_t-\bfP_s}_1$ each, the middle one at most
$\varepsilon_{\rm opt}$.
\end{proof}

\subsection{Threshold Policies and Regret}

\begin{theorem}[Regret of the TV-threshold policy]
\label{thm:regret}
Consider the policy that, with $s$ the time of the last reconfiguration (initially $s=1$,
with an $\varepsilon_{\rm opt}$-optimal initial wiring), reconfigures at $t$ to an
$\varepsilon_{\rm opt}$-optimal wiring for $F_t$ if and only if
$\norm{\bfP_t-\bfP_s}_1>\vartheta$.  Then $N_T<V_1/\vartheta$ and
\begin{equation}
  R_T\le T\,(2\Gamma\vartheta+\varepsilon_{\rm opt})+\frac{C\,V_1}{\vartheta}.
\end{equation}
With $\vartheta^\star=\sqrt{CV_1/(2\Gamma T)}$,
$R_T\le2\sqrt{2\Gamma\,C\,V_1\,T}+T\varepsilon_{\rm opt}$.
\end{theorem}

\begin{proof}
At reconfiguration times the excess is at most $\varepsilon_{\rm opt}$; otherwise
$\norm{\bfP_t-\bfP_s}_1\le\vartheta$ and \cref{lem:holding} bounds the excess by
$2\Gamma\vartheta+\varepsilon_{\rm opt}$.  Let $1=t_0<t_1<\dots<t_{N}$ be the reconfiguration
times.  By the triangle inequality for $\norm{\cdot}_1$,
$\vartheta<\norm{\bfP_{t_j}-\bfP_{t_{j-1}}}_1\le\sum_{u=t_{j-1}+1}^{t_j}\norm{\bfP_u-\bfP_{u-1}}_1$;
the ranges are disjoint, so summing over $j$ gives $N\vartheta<V_1$.  The choice of
$\vartheta^\star$ balances the two terms.
\end{proof}

\begin{remark}[On the comparator]
The benchmark in $R_T$ is the dynamic optimum that pays nothing for switching; it is a
demanding comparator, and $R_T$ therefore also contains the unavoidable cost of tracking a
moving optimum.  The theorem does not compare the policy with the best policy that accounts
for switching costs, and it treats the computation of an $\varepsilon_{\rm opt}$-optimal wiring
as instantaneous: computation time enters only through $\varepsilon_{\rm opt}$ (a wiring
computed with a limited iteration budget) and, if it causes downtime, through $C$.
\end{remark}

The counting argument rests on the triangle inequality, which total variation satisfies and
the KL divergence does not.  KL statistics can nevertheless be employed by means of Pinsker's
inequality and its reverse, provided the reference distribution is bounded away from zero.

\begin{corollary}[KL-threshold policy]
\label{cor:kl}
Assume that every reference distribution $\bfP_s$ used by the policy satisfies
$\min_iP_{s,i}\ge p_{\min}>0$, and that the policy reconfigures at $t$ if and only if
$\KL{\bfP_t}{\bfP_s}>\vartheta^2/2$.  Then
\begin{equation}
  N_T<\frac{V_1}{\vartheta\,(p_{\min}/2)^{1/2}}\qquad\text{and}\qquad
  R_T\le T\,(2\Gamma\vartheta+\varepsilon_{\rm opt})+\frac{C\,V_1}{\vartheta\,(p_{\min}/2)^{1/2}}.
\end{equation}
\end{corollary}

\begin{proof}
\emph{(i) Holding steps.}  If no reconfiguration occurs at $t$, then
$\KL{\bfP_t}{\bfP_s}\le\vartheta^2/2$, and Pinsker's inequality gives
$\norm{\bfP_t-\bfP_s}_1\le\sqrt{2\KL{\bfP_t}{\bfP_s}}\le\vartheta$; by \cref{lem:holding} the excess is at
most $2\Gamma\vartheta+\varepsilon_{\rm opt}$.

\emph{(ii) Reconfiguration steps.}  Since $\log(1+x)\le x$,
$\KL{\bfP_t}{\bfP_s}\le\log\bigl(1+\chi^2(\bfP_t\|\bfP_s)\bigr)\le\chi^2(\bfP_t\|\bfP_s)$, where
$\chi^2(\bfP_t\|\bfP_s)=\sum_i(P_{t,i}-P_{s,i})^2/P_{s,i}\le p_{\min}^{-1}\sum_i(P_{t,i}-P_{s,i})^2
\le p_{\min}^{-1}\norm{\bfP_t-\bfP_s}_1^2$.  Hence a reconfiguration implies
$\vartheta^2/2<p_{\min}^{-1}\norm{\bfP_t-\bfP_s}_1^2$, that is,
$\norm{\bfP_t-\bfP_s}_1>\vartheta\,(p_{\min}/2)^{1/2}$.

\emph{(iii) Counting.}  The argument of \cref{thm:regret}, with the threshold
$\vartheta\,(p_{\min}/2)^{1/2}$ in place of $\vartheta$, gives the bound on $N_T$; combining it with
(i) yields the bound on $R_T$.
\end{proof}

Only the reference distributions need to be bounded away from zero; in practice this is
ensured by additive smoothing of the reference estimate.

In practice $\bfP_t$ is estimated from $n$ sampled packets per iteration.  We first record a
deviation bound whose dimension dependence is through the support size only.

\begin{lemma}[$\ell_1$ deviation of the empirical distribution]
\label{lem:l1dev}
Let $\hat\bfP$ be the empirical distribution of $n$ i.i.d.\ samples from $\bfP$, whose support
has $m_{\rm eff}$ cells.  Then $\E\norm{\hat\bfP-\bfP}_1\le\sqrt{m_{\rm eff}/n}$ and, for every
$u>0$, $\Pr\bigl(\norm{\hat\bfP-\bfP}_1\ge\sqrt{m_{\rm eff}/n}+u\bigr)\le e^{-nu^2/2}$.
\end{lemma}

\begin{proof}
For a cell $i$, $\E|\hat P_i-P_i|\le(\mathrm{Var}\,\hat P_i)^{1/2}=\sqrt{P_i(1-P_i)/n}\le\sqrt{P_i/n}$;
cells outside the support contribute nothing.  By the Cauchy--Schwarz inequality,
$\sum_i\sqrt{P_i}\le\sqrt{m_{\rm eff}\sum_iP_i}=\sqrt{m_{\rm eff}}$, which gives the bound on the
mean.  Replacing one of the $n$ samples changes $\norm{\hat\bfP-\bfP}_1$ by at most $2/n$, so
McDiarmid's bounded-differences inequality \cite{mcdiarmid1989method} gives
$\Pr(\norm{\hat\bfP-\bfP}_1-\E\norm{\hat\bfP-\bfP}_1\ge u)\le\exp\bigl(-2u^2/(n\cdot4/n^2)\bigr)=e^{-nu^2/2}$.
\end{proof}

\begin{theorem}[Sampled traffic]
\label{thm:regret_sampled}
Let $\hat\bfP_t$ be the empirical distribution of $n$ i.i.d.\ samples from $\bfP_t$, let
$m_{\rm eff}$ bound the support sizes of all $\bfP_t$, and let the policy reconfigure (to a wiring
$\varepsilon_{\rm opt}$-optimal for the \emph{estimated} workload $\hat\bfP_t$) when
$\norm{\hat\bfP_t-\hat\bfP_s}_1>\vartheta$.  For $\delta\in(0,1)$ put
\begin{equation}
  \varepsilon=\sqrt{\frac{m_{\rm eff}}{n}}+\sqrt{\frac{2\log(T/\delta)}{n}}.
  \label{eq:eps_n}
\end{equation}
Then, with probability at least $1-\delta$, if $\vartheta>2\varepsilon$,
\begin{equation}
  R_T\le T\bigl(2\Gamma(\vartheta+\varepsilon)+\varepsilon_{\rm opt}\bigr)+\frac{C\,V_1}{\vartheta-2\varepsilon}.
\end{equation}
\end{theorem}

\begin{proof}
By \cref{lem:l1dev} with $u=\sqrt{2\log(T/\delta)/n}$,
$\Pr(\norm{\hat\bfP_t-\bfP_t}_1\ge\varepsilon)\le\delta/T$ for each $t$; a union bound gives the event
$\mathcal E=\{\norm{\hat\bfP_t-\bfP_t}_1<\varepsilon\ \forall t\}$ with probability at least $1-\delta$.  On $\mathcal E$: the wiring deployed at $s$ is $\varepsilon_{\rm opt}$-optimal for
$\hat\bfP_s$, so by \cref{lem:holding} applied with $\hat\bfP_s$ in place of $\bfP_s$,
$F_t(\bfw_s)-F_t^\star\le2\Gamma\norm{\bfP_t-\hat\bfP_s}_1+\varepsilon_{\rm opt}
\le2\Gamma(\norm{\bfP_t-\hat\bfP_t}_1+\norm{\hat\bfP_t-\hat\bfP_s}_1)+\varepsilon_{\rm opt}
\le2\Gamma(\varepsilon+\vartheta)+\varepsilon_{\rm opt}$ whenever no reconfiguration occurs (and
$\le2\Gamma\varepsilon+\varepsilon_{\rm opt}$ when it does).  A reconfiguration implies
$\norm{\bfP_t-\bfP_s}_1\ge\norm{\hat\bfP_t-\hat\bfP_s}_1-2\varepsilon>\vartheta-2\varepsilon$, and the
counting argument gives $N_T<V_1/(\vartheta-2\varepsilon)$.
\end{proof}

The radius \eqref{eq:eps_n} depends on the number of cells only through the support size
$m_{\rm eff}$, not through the ambient number $M$ of cells.  In the experiment of
\cref{sec:experiments} ($n=2\times10^5$, $T=600$, $\delta=0.05$, $m_{\rm eff}\le120$ occupied cells)
it gives $\varepsilon\approx0.034$, so the guarantee is non-vacuous for thresholds
$\vartheta>0.07$.

\subsection{The KL Indicator as a Likelihood-Ratio Statistic}

Let $\bfP_{\rm ref}$ be the reference distribution at the last reconfiguration and let
$\bm n=(n_1,\dots,n_M)$ be cell counts of $N$ samples with empirical distribution $\hat\bfP$.

\begin{proposition}[GLR identity]
\label{prop:glr}
For the multinomial model, the log generalised likelihood ratio of
$H_0:\bfP=\bfP_{\rm ref}$ against $H_1:\bfP\in\Simp$ equals
\begin{equation}
  \log\frac{\sup_{\bfP\in\Simp}\prod_iP_i^{n_i}}{\prod_iP_{{\rm ref},i}^{n_i}}
  =N\,\KL{\hat\bfP}{\bfP_{\rm ref}}.
\end{equation}
\end{proposition}

\begin{proof}
The supremum is attained at the maximum-likelihood estimate $\hat\bfP=\bm n/N$, and
\[
  \textstyle\sum_in_i\log(\hat P_i/P_{{\rm ref},i})=N\sum_i\hat P_i\log(\hat P_i/P_{{\rm ref},i}).
\]
\end{proof}

Hence a KL indicator of traffic change acquires a precise statistical meaning once its
reference is the distribution at the time of configuration and it is weighted by the sample
size: it is then the log generalised likelihood ratio.  The identity is classical; its use
here is to fix the reference and the weighting that a principled detector must employ.  If a change may have occurred at an unknown time
$\tau\in(s,t]$, the GLR statistic maximises over $\tau$
\cite{lorden1971procedures,siegmund1995glr}:
\begin{equation}
  \Lambda_t=\max_{s<\tau\le t}N_{\tau:t}\,\KL{\hat\bfP_{\tau:t}}{\bfP_{\rm ref}},\qquad
  \hat a(t)=t-\hat\tau+1,
  \label{eq:glrcusum}
\end{equation}
where $\hat\bfP_{\tau:t}$ pools the $N_{\tau:t}$ samples of iterations $\tau,\dots,t$ and
$\hat\tau$ is the maximiser.  The \emph{lag} $\hat a(t)$ is thus selected by maximum
likelihood rather than by a heuristic, and the pooled estimate $\hat\bfP_{\hat\tau:t}$ is the
natural input for re-optimisation.

\begin{theorem}[False-alarm control]
\label{thm:falsealarm}
Suppose that no change occurs, i.e.\ all samples in iterations $s+1,\dots,t$ are i.i.d.\ from
$\bfP_{\rm ref}$, and the maximisation in \eqref{eq:glrcusum} ranges over at most $W$ windows,
each with at most $N_{\max}$ samples.  Then for every $h>0$,
\begin{equation}
  \Pr(\Lambda_t\ge h)\le W\,(N_{\max}+1)^{M}e^{-h}.
\end{equation}
In particular the threshold $h=M\log(N_{\max}+1)+\log(TW/\delta)$ yields at most $\delta$
probability of any false alarm over $T$ iterations.
\end{theorem}

\begin{proof}
Fix a window with $N$ samples.  By the method of types \cite{cover2006elements,csiszar1998types},
the probability of observing a particular type $Q$ is at most $e^{-N\KL{Q}{\bfP_{\rm ref}}}$,
and there are at most $(N+1)^M$ types.  Summing over types with $N\KL{Q}{\bfP_{\rm ref}}\ge h$
gives $\Pr(N\KL{\hat\bfP}{\bfP_{\rm ref}}\ge h)\le(N+1)^Me^{-h}$.  A union bound over windows
(and, for the second claim, over iterations) completes the proof.
\end{proof}

\begin{remark}
The bound is non-asymptotic but conservative when $M$ is large; asymptotically
$2N\KL{\hat\bfP}{\bfP_{\rm ref}}$ is $\chi^2_{M'-1}$-distributed under $H_0$ (Wilks), with $M'$ the
support size, which gives practical thresholds.  When $\bfP_{\rm ref}$ is itself estimated we
use additive smoothing; we do not analyse the resulting plug-in effect.  The sequential
statistic \eqref{eq:glrcusum} and the regret analysis are complementary: the former
controls false alarms, the latter quantifies the cost of any threshold.
\end{remark}

\section{Multi-Tenant Clusters}
\label{sec:multitenant}

\subsection{Shared Port Budget and the KL Chain Rule}

Let $K$ jobs occupy disjoint pod sets $V^{(1)},\dots,V^{(K)}$.  Since there is no traffic
between jobs, the candidate circuits are $\Eocs=\bigsqcup_k\Eocs_k$ with $\Eocs_k$ the pairs
inside $V^{(k)}$.  The jobs share the OCS port budget $\mu<\sum_k|V^{(k)}|r/2$ through the
simplex constraint, and $\calW$ is the degree polytope \eqref{eq:calW}.  Write
$\rho_k=\sum_{e\in\Eocs_k}w_e$ (the share of job $k$) and $w^{(k)}_e=w_e/\rho_k$ for
$e\in\Eocs_k$, and similarly $\bfq\mapsto(\bfrho^{q},\bfq^{(k)})$.

\begin{lemma}[Chain rule]
\label{lem:chain}
$\KL{\bfw}{\bfq}=\KL{\bfrho}{\bfrho^q}+\sum_{k=1}^K\rho_k\,\KL{\bfw^{(k)}}{\bfq^{(k)}}$.
\end{lemma}

\begin{proof}
Expand $\sum_k\sum_{e\in\Eocs_k}\rho_kw^{(k)}_e\log\frac{\rho_kw^{(k)}_e}{\rho^q_kq^{(k)}_e}$ and use
$\sum_{e\in\Eocs_k}w^{(k)}_e=1$.
\end{proof}

The prior thus penalises separately the allocation of the budget across jobs and the shape
of each job's wiring, corresponding to the mixed coordinates
$(\bfrho,\{\bfw^{(k)}\})$ of the simplex.

\subsection{Fairness Functionals}

Let $J_k(\bfw)$ be the smoothed iteration time of job $k$ (it depends only on the block
$\Eocs_k$), $J_k^{\rm ref}>0$ a reference time (e.g.\ job $k$ alone with all its ports), and
$s_k=J_k/J_k^{\rm ref}$ the \emph{slowdown}.  Following the fairness families of network resource allocation \cite{mo2000fair},
for weights $\omega\in\Simp^{K-1}$, $\omega>0$, and $p\in[1,\infty)$ define
\begin{equation}
  \Phi_p(\bm s)=\Bigl(\sum_k\omega_ks_k^p\Bigr)^{1/p},\qquad
  \Phi_\infty(\bm s)=\max_ks_k\ \ (\text{smoothed by }\lse_{\tau_f}),
\end{equation}
and the multi-tenant free energy $F^{(K)}(\bfw)=\Phi_p(\bm s(\bfw))+\lambda\KL{\bfw}{\bfq}$.

\begin{proposition}[Convexity; relation to $\alpha$-fairness]
\label{prop:fairconvex}
For $p\in[1,\infty]$, $\bfw\mapsto\Phi_p(\bm s(\bfw))$ is convex on $\calW$, and it is
$L$-smooth in $\ell_1$ for some finite $L$ when $p<\infty$ or when $\Phi_\infty$ is smoothed,
so \cref{thm:wellposed,thm:convergence} apply to $F^{(K)}$.

Moreover, let the normalised throughput of job $k$ be $x_k(\bfw)=1/s_k(\bfw)$ and
$u_\alpha(x)=x^{1-\alpha}/(1-\alpha)$ ($\alpha\ne1$).  For $\alpha=p+1\ge2$ the identity
\begin{equation}
  \sum_k\omega_k\,u_{p+1}\bigl(x_k(\bfw)\bigr)=-\frac1p\,\Phi_p\bigl(\bm s(\bfw)\bigr)^p
  \qquad(\bfw\in\calW)
  \label{eq:alpha_identity}
\end{equation}
holds.  Consequently, for $\lambda=0$ the set of wirings maximising the $\alpha$-fair utility of
throughputs over $\calW$ coincides with the set minimising $\Phi_p\circ\bm s$; the two
objectives are related by a strictly decreasing transformation and their values differ.  For
$\lambda>0$, the problems $\min\,\Phi_p\circ\bm s+\lambda\KL{\cdot}{\bfq}$ and
$\min\,\frac1p(\Phi_p\circ\bm s)^p+\lambda\KL{\cdot}{\bfq}$ are both convex but have different
minimisers in general.
\end{proposition}

\begin{proof}
Each $s_k$ is convex and positive (\cref{lem:smooth}); $\Phi_p$ is a weighted $p$-norm on
$\R^K_{\ge0}$, hence convex and nondecreasing in each argument, and so is $\lse_{\tau_f}$; the
composition of a convex nondecreasing function with convex functions is convex.  On the
compact set $\calW$ the slowdowns lie in $[s_{\min},s_{\max}]$ with $s_{\min}>0$, where $\Phi_p$
is $C^2$; the chain rule and \cref{lem:smooth} bound the Hessian of the composition, giving a
finite $L$.  For \eqref{eq:alpha_identity}, $u_{p+1}(1/s)=(1/s)^{-p}/(-p)=-s^p/p$, and summing
with weights $\omega_k$ gives $-\frac1p\sum_k\omega_ks_k^p=-\frac1p\Phi_p^p$.  Since $t\mapsto-t^p/p$
is strictly decreasing on $\R_{>0}$, maximisers of the left-hand side are exactly minimisers of
$\Phi_p$.  With $\lambda>0$ the transformation acts on one term of the objective only, so the
minimisers need not coincide; convexity of $\frac1p\Phi_p^p=\frac1p\sum_k\omega_ks_k^p$ follows as
above since $t\mapsto t^p$ is convex and nondecreasing on $\R_{\ge0}$.
\end{proof}

\begin{remark}[The range $\alpha\ge2$]
Two cautions apply.  First, $\alpha$-fair utilities must be applied to throughputs, not to
costs: $U\mapsto U^{1-\alpha}/(1-\alpha)$ has second derivative $-\alpha U^{-\alpha-1}<0$ and is
therefore concave in the cost.  Second, for $\alpha<2$ the $\alpha$-fair problem in throughputs
need not be convex in the wiring: with the convex cost $J(w)=1+w^2$, the throughput $1/J$ fails
to be concave at $w=1$.  \Cref{prop:fairconvex} thus singles out $\alpha\ge2$ as the range in
which fairness and convexity are compatible for this model.
\end{remark}

\begin{theorem}[Gibbs form of the fair optimum]
\label{thm:kkt}
Let $\bfw^\star$ minimise $F^{(K)}$ over $\calW$ and let
$\pi_k=\partial\Phi_p/\partial s_k$ at $\bm s(\bfw^\star)$ (for $p<\infty$,
$\pi_k=\omega_ks_k^{p-1}\Phi_p^{1-p}$).  Then $\bfw^\star>0$ and there exist $\bfy^\star\in\R^V_{\ge0}$
with $y^\star_v(A\bfw^\star-(r/\mu)\one)_v=0$ such that, for $e\in\Eocs_k$,
\begin{equation}
  w^\star_e=\frac1Z\,q_e\exp\Bigl(-\frac{\pi_k\,\partial_es_k(\bfw^\star)+(A^\top\bfy^\star)_e}{\lambda}\Bigr).
  \label{eq:gibbs}
\end{equation}
Conversely, any $\bfw\in\calW$ satisfying \eqref{eq:gibbs} with such $\bfy^\star$ is the minimiser.
\end{theorem}

\begin{proof}
Positivity follows from \cref{lem:positivity}.  The constraints are affine, so the KKT
conditions are necessary and sufficient for the convex problem (Abadie's constraint
qualification holds for affine constraints).  Stationarity at the positive point reads
$\nabla(\Phi_p\circ\bm s)(\bfw^\star)+\lambda(\log(\bfw^\star/\bfq)+\one)+A^\top\bfy^\star+\nu\one=0$,
with the multipliers of $\bfw\ge0$ vanishing by complementary slackness.  The chain rule
gives $\partial_e(\Phi_p\circ\bm s)=\pi_k\partial_es_k$ for $e\in\Eocs_k$, since $s_j$ does
not depend on $w_e$ for $j\ne k$.  Solving for $\bfw^\star$ gives \eqref{eq:gibbs}.
\end{proof}

The fairness weights $\pi_k\propto\omega_ks_k^{p-1}$ increase with the slowdown, so jobs that
are worse off pull the shared wiring towards themselves, increasingly so as $p$ grows; as
$p\to\infty$ the weight concentrates on the most slowed-down jobs.

\begin{proposition}[Pareto optimality]
\label{prop:pareto}
If $\omega>0$ and $p<\infty$, the minimiser of $F^{(K)}$ is Pareto optimal for the vector
$\bigl(s_1,\dots,s_K,\KL{\cdot}{\bfq}\bigr)$ over $\calW$; for $p=\infty$ it is weakly Pareto
optimal.
\end{proposition}

\begin{proof}
For $p<\infty$, $\Phi_p$ is strictly increasing in each coordinate on $\R^K_{>0}$; a point
dominating $\bfw^\star$ would have strictly smaller $F^{(K)}$.  For $p=\infty$ only weak
monotonicity holds.
\end{proof}

\begin{remark}[Structure of the Pareto set]
The Pareto set is not, in general, the convex hull of the single-job optima: the shape
$\bfw^{(k)}$ that is optimal for job $k$ depends on its share $\rho_k$ through the degree
constraints, so mixtures of single-job optima computed with the full budget are typically
infeasible or dominated.
\end{remark}

\subsection{Reconfiguration with Externalities}

Reconfiguring job $k$ only within its block, i.e.\ changing $\bfw^{(k)}$ with $\bfrho$ fixed,
touches no circuit of any other job, because blocks are disjoint.  For such
\emph{block-restricted} reconfigurations \cref{thm:regret,thm:regret_sampled} apply job by job
(with constants $\Gamma_k$, $C_k$, $V_{1,k}$), relative to the dynamic optimum with the shares
held fixed; the regrets add.  Changing $\bfrho$ moves circuits between jobs and disturbs all
jobs whose blocks change; applying \cref{thm:regret} to the aggregate objective with the
social cost $C_{\rm soc}=\sum_kC_k\mathbf 1[\text{block }k\text{ changes}]$ internalises
this externality.

\section{Synthetic Experiments}
\label{sec:experiments}

\subsection{Setup}

All experiments use the model of \cref{sec:model}, implemented in Python/NumPy; the KL
projection is computed through the dual of \cref{lem:dual} with a bound-constrained
quasi-Newton method, global optima of the continuous problem are cross-checked with CVXPY
and the Clarabel conic solver, and integer optima are computed with the HiGHS MILP solver
through SciPy.  \Cref{tab:setup} lists the parameters.  Volumes are per iteration and
aggregated over micro-batches; they are chosen to make pipeline, all-reduce and
all-to-all phases comparable, not to reproduce a particular model.

\begin{table}[h]
\centering
\caption{Default experimental parameters.}
\label{tab:setup}
\small
\renewcommand{\arraystretch}{1.15}
\begin{tabular}{lp{12.6cm}}
\toprule
Fabric & $N=16$ pods, $r=8$ OCS ports/pod, $c=50$\,GB/s per circuit, $b_0=10$\,GB/s,
$\mu=56$ ($7/8$ of the ports)\\
Job & $S=4$ pipeline stages $\times$ $D=4$ data-parallel replicas, one pod per (stage, replica)\\
Compute & forward 80\,ms, backward 160\,ms per stage, optimiser 20\,ms\\
Volumes & PP 3\,GB per stage boundary and replica; ring all-reduce 10\,GB per ring edge;
 MoE all-to-all $4\,\mathrm{GB}\times\max(\pi_j,\pi_{j'})$ between replicas $j,j'$ of a stage,
 expert popularity $\pi=(0.4,0.3,0.2,0.1)$ unless stated\\
DAG & forward chain, backward chain, all-reduce of stage $i$ after its backward pass\\
Smoothing & $\tau=\tau_s=1$\,ms; prior $\bfq$ uniform over all $\binom{16}{2}=120$ pairs\\
\bottomrule
\end{tabular}
\end{table}

\subsection{Convergence and Temperature}

\Cref{fig:convergence} runs \cref{alg:main} from $\bfq$ for three temperatures.  The free
energy decreases monotonically in every run, and the KL distance to the optimum stays below
the bound of \cref{thm:convergence}(c) evaluated with the step sizes $\eta_k$ actually
accepted by backtracking (dotted).  The bound is valid but conservative: the accepted steps
are small (median $\eta_k$ between $1.6\times10^{-2}$ and $6.3\times10^{-2}$) because the
curvature of the smoothed maxima is large, while the observed contraction is much faster.
\Cref{tab:conv} reports the checks; the final iterate agrees with the CVXPY/Clarabel
optimum to within solver accuracy wherever Clarabel returned an optimal status.

\begin{figure}[t]
\centering
\includegraphics[width=\linewidth]{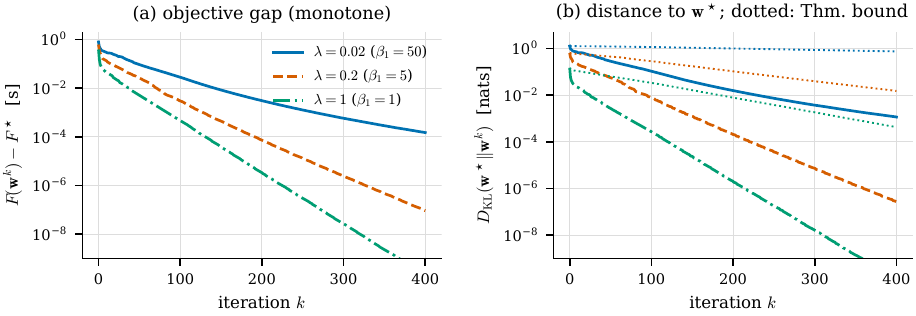}
\caption{Convergence of \cref{alg:main} for $\lambda=1/\beta_1\in\{0.02,0.2,1\}$.
(a) Objective gap (log scale).  (b) $\KL{\bfw^\star}{\bfw^k}$ (solid) and the bound
$\prod_{i<k}(1+\eta_i\lambda)^{-1}\KL{\bfw^\star}{\bfq}$ of \cref{thm:convergence}(c) (dotted,
same colour).}
\label{fig:convergence}
\end{figure}

\begin{table}[h]
\centering
\caption{Checks of \cref{thm:convergence} (400 iterations from $\bfq$).  ``Bound viol.'' is the
largest excess of $\KL{\bfw^\star}{\bfw^k}$ over its bound; the last column is the difference to
the CVXPY/Clarabel optimum.}
\label{tab:conv}
\small
\begin{tabular}{lccccc}
\toprule
& $\beta_1$ & monotone & bound viol. & final gap [s] & $|F_{\rm MD}-F_{\rm CVX}|$\\
\midrule
$\lambda=0.02$ & 50 & yes & $0$ & $1.5\times10^{-4}$ & $4.4\times10^{-9}$\\
$\lambda=0.2$ & 5 & yes & $0$ & $9.3\times10^{-8}$ & solver failed\\
$\lambda=1$ & 1 & yes & $0$ & $2.2\times10^{-10}$ & $<10^{-12}$\\
\bottomrule
\end{tabular}
\end{table}

\Cref{fig:temperature} sweeps $\beta_1$ from $0.1$ to $10^4$.  As predicted by
\cref{thm:limits}, $U(\bfw^\star_{\beta_1})$ decreases and $\KL{\bfw^\star_{\beta_1}}{\bfq}$ increases
monotonically, between the static value $U(\bfq)$ and $\min_\calW U$; both quantitative
bounds hold at every grid point, and so does the Jeffreys bound of \cref{thm:jeffreys}.  The
hard iteration time falls from \Jstatic\,s for the static symmetric wiring to \Jadapt\,s for the
fully adaptive one (\Jgain\% lower).  Most of the gain is obtained for $\beta_1\gtrsim10$, at a KL
departure of about one nat; the curve thus indicates how much asymmetry a given gain
requires, which is the design information $\beta_1$ is meant to provide.

\begin{figure}[t]
\centering
\includegraphics[width=\linewidth]{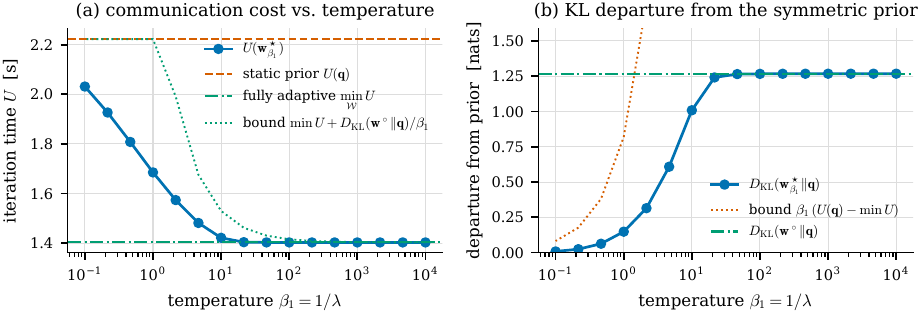}
\caption{Temperature sweep.  (a) Smoothed iteration time of $\bfw^\star_{\beta_1}$ with the
static and fully adaptive levels and the bound of \cref{thm:limits}(iii).  (b) KL departure
from the prior with the bound of \cref{thm:limits}(ii) and its limit
$\KL{\bfw^\circ}{\bfq}$.}
\label{fig:temperature}
\end{figure}

\subsection{Integer Circuits Against an Exact MILP}
\label{sec:exp_integer}

We compare, for $N\in\{16,24,32\}$ pods ($S\in\{4,6,8\}$, $D=4$), dense and MoE workloads,
and $r\in\{4,8,16,32\}$ ports per pod at fixed pod bandwidth $rc=200$\,GB/s ($\mu=Nr/2$):
(i)~\cref{alg:main} (300 iterations, $\lambda=10^{-3}$) followed by floor rounding with greedy
filling, (ii)~the same followed by the local search of \cref{sec:rounding}, and (iii)~the exact
MILP, in which each pair's circuit count is encoded by binary variables and transfer times
by their exact values at each count.  Results are normalised by the relaxation bound
$\min_{\calW'}J$ of \cref{prop:rounding} (computed with CVXPY).

\begin{figure}[t]
\centering
\includegraphics[width=\linewidth]{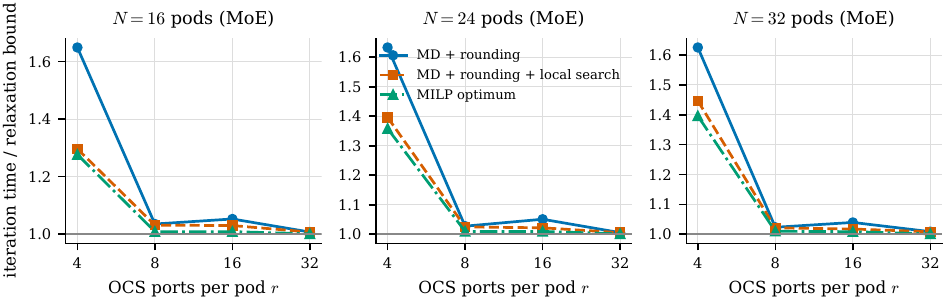}
\caption{Integer wirings for MoE workloads: iteration time divided by the relaxation lower
bound, as a function of the number of OCS ports per pod at fixed pod bandwidth.}
\label{fig:rsweep}
\end{figure}

\Cref{fig:rsweep,tab:integer} lead to three observations, which we report without
qualification.  First, the exact MILP is fast on all instances (at most 4.5\,s with HiGHS)
and always finds the best integer wiring; on these sizes the continuous method therefore does
\emph{not} replace integer programming.  Second, the relaxation bound of
\cref{prop:rounding} is tight once pods have several ports: for $r\ge8$ the MILP optimum is
within 3.3\% of the bound (within 0.6\% for $r=32$), so the continuous problem solved by
\cref{alg:main} accurately predicts the achievable iteration time, whereas for $r=4$ the
integrality gap reaches 40\%.  Third, plain rounding loses up to 37\% at $r\le8$, and the
local search recovers most of it: \cref{alg:main} with rounding and local search is within
1.5--14.1\% of the MILP for $r=4$ and within 0--4.4\% for $r\ge8$, the gap shrinking with $r$
as \cref{prop:rounding} predicts.  The running times of \cref{alg:main} (14--57\,s for 300
iterations) and especially of the local search reflect an unoptimised Python
implementation; they are not competitive with a compiled MILP solver at this scale, and we
make no runtime claim.

\begin{table}[t]
\centering
\caption{Integer wirings.  LB: relaxation bound; Rnd: \cref{alg:main} + rounding; LS: + local
search; MILP: exact (opt.\ = proven optimal).  Times in seconds (single core, Python
implementation for MD and LS).}
\label{tab:integer}
\scriptsize
\setlength{\tabcolsep}{4pt}
\begin{tabular}{rlrrrrrrrr}
\toprule
$N$ & workload & $r$ & LB & Rnd & LS & MILP & $t_{\rm MD}$ & $t_{\rm LS}$ & $t_{\rm MILP}$\\
\midrule
16 & dense & 4 & 1.307 & 1.979 & 1.650 & 1.447 (opt.) & 26.4 & 9.4 & 0.09\\
16 & dense & 8 & 1.307 & 1.800 & 1.358 & 1.339 (opt.) & 40.3 & 12.0 & 0.27\\
16 & dense & 16 & 1.307 & 1.326 & 1.314 & 1.314 (opt.) & 27.1 & 18.4 & 0.35\\
16 & dense & 32 & 1.307 & 1.334 & 1.313 & 1.313 (opt.) & 44.7 & 28.4 & 1.97\\
24 & dense & 4 & 1.898 & 2.568 & 2.239 & 2.064 (opt.) & 37.5 & 18.3 & 0.22\\
24 & dense & 8 & 1.898 & 2.395 & 2.006 & 1.960 (opt.) & 14.1 & 26.8 & 0.63\\
24 & dense & 16 & 1.898 & 2.003 & 2.003 & 1.919 (opt.) & 14.5 & 26.9 & 1.56\\
24 & dense & 32 & 1.898 & 2.010 & 1.913 & 1.908 (opt.) & 27.5 & 98.9 & 4.51\\
32 & dense & 4 & 2.487 & 3.157 & 2.828 & 2.670 (opt.) & 44.2 & 40.2 & 0.22\\
32 & dense & 8 & 2.487 & 2.984 & 2.605 & 2.549 (opt.) & 39.1 & 102.9 & 0.87\\
32 & dense & 16 & 2.487 & 2.599 & 2.597 & 2.508 (opt.) & 27.7 & 110.4 & 3.18\\
32 & dense & 32 & 2.487 & 2.503 & 2.498 & 2.498 (opt.) & 24.7 & 230.7 & 4.18\\
16 & MoE & 4 & 1.692 & 2.790 & 2.193 & 2.160 (opt.) & 57.4 & 22.1 & 0.12\\
16 & MoE & 8 & 1.692 & 1.751 & 1.745 & 1.707 (opt.) & 20.9 & 33.4 & 0.05\\
16 & MoE & 16 & 1.692 & 1.781 & 1.742 & 1.707 (opt.) & 26.0 & 34.7 & 1.02\\
16 & MoE & 32 & 1.692 & 1.704 & 1.704 & 1.694 (opt.) & 18.5 & 47.2 & 0.64\\
24 & MoE & 4 & 2.517 & 4.110 & 3.513 & 3.418 (opt.) & 24.3 & 53.6 & 0.29\\
24 & MoE & 8 & 2.517 & 2.585 & 2.578 & 2.540 (opt.) & 15.0 & 128.5 & 0.11\\
24 & MoE & 16 & 2.517 & 2.644 & 2.571 & 2.540 (opt.) & 20.8 & 175.8 & 2.01\\
24 & MoE & 32 & 2.517 & 2.530 & 2.530 & 2.519 (opt.) & 16.9 & 194.3 & 1.11\\
32 & MoE & 4 & 3.341 & 5.430 & 4.833 & 4.664 (opt.) & 22.3 & 175.9 & 0.85\\
32 & MoE & 8 & 3.341 & 3.418 & 3.412 & 3.374 (opt.) & 31.4 & 316.3 & 0.13\\
32 & MoE & 16 & 3.341 & 3.473 & 3.399 & 3.369 (opt.) & 36.0 & 395.2 & 2.78\\
32 & MoE & 32 & 3.341 & 3.369 & 3.368 & 3.345 (opt.) & 34.6 & 584.4 & 1.77\\
\bottomrule
\end{tabular}
\end{table}

\subsection{Dynamic Reconfiguration}

The MoE expert popularity of each stage follows a random walk on its logits (standard
deviation $0.03$ per iteration) with abrupt routing shifts at $t=120$ and $t=440$; at $t=280$
the job is re-placed onto a random permutation of the pods.  The workload distribution is
defined on the $M=18\times120$ cells (phase slot, pair); each iteration $n=2\times10^5$ cell
samples are observed.  A reconfiguration costs $C=1$\,s and re-optimises by 40 iterations of
\cref{alg:main} on the estimated workload, started from $\frac12(\bfw_{\rm current}+\bfq)$.
Mixing with the prior matters: multiplicative updates recover slowly from coordinates that
are nearly zero (after a re-placement the previously unused pairs are needed), and since
$\bfw^0\ge\bfq/2$ coordinatewise, $\KL{\bfw^\star}{\bfw^0}\le\log2+\KL{\bfw^\star}{\bfq}$, so the
bound of \cref{thm:convergence} does not depend on the previous wiring.  The dynamic oracle
re-optimises the true workload every iteration at no cost (25 warm-started iterations, and
400 iterations from both $\bfq$ and the previous wiring at $t=0$ and at the three change points).  We compare the static prior, a wiring optimised
once, periodic re-optimisation, the TV-threshold policy of \cref{thm:regret_sampled}, and the
GLR--CUSUM detector \eqref{eq:glrcusum} with windows of at most 30 iterations, where the
re-optimisation uses the pooled estimate $\hat\bfP_{\hat\tau:t}$.

\Cref{tab:dynamic,fig:dynamic} summarise the results (workload path length $V_1=\DynVone$).
Optimising once and never reconfiguring is the worst policy (regret \NeverRegret\,s): after the
re-placement at $t=280$ the old wiring serves pairs that no longer carry traffic, and it is
then far worse than even the static prior (\StaticRegret\,s), which illustrates why a
symmetric fallback is valuable.  Every adaptive family shows the trade-off of
\cref{thm:regret} between holding excess and switching cost, with an interior optimum
(\cref{fig:dynamic}(b)).  The best periodic schedule ($k=10$) attains \BestPer\,s; the best
TV threshold ($\vartheta=0.05$) attains \BestTV\,s with 9 reconfigurations; the GLR--CUSUM
detector attains the lowest regret, \BestGLR\,s with 16 reconfigurations at $h=10^4$, and is
within 4\% of it over $h\in[10^4,3\times10^4]$.  The conservative threshold of
\cref{thm:falsealarm} for $\delta=0.05$, $h=M\log(N_{\max}+1)+\log(TW/\delta)\approx3.4\times10^4$, lies
in this good range.  Two caveats are visible.  Large TV thresholds ($\vartheta\ge0.2$) react only
to the re-placement and miss the routing shifts, and very small GLR thresholds
($h\le10^3$) fire at every iteration, because sampling noise in the estimated reference
already produces statistics of that size.  Moreover, all policies retain an excess of
roughly 0.1\,s per iteration after re-optimisation (\cref{fig:dynamic}(a)); this is the
$\varepsilon_{\rm opt}$ and estimation term of \cref{thm:regret_sampled}, caused by the limited
budget of 40 iterations per reconfiguration and by the pooled estimate, and it bounds what any
detector can achieve here.

\begin{figure}[t]
\centering
\includegraphics[width=\linewidth]{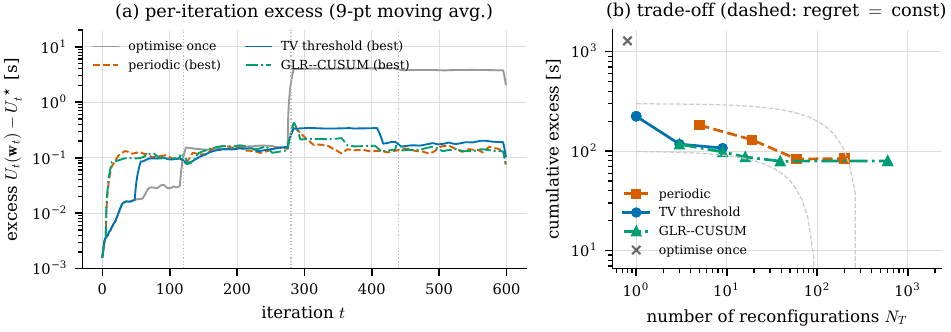}
\caption{Dynamic workload ($T=\DynT$, $C=1$\,s).  (a) Excess iteration time over the dynamic
oracle (log scale, 9-point moving average; vertical lines: abrupt changes).  (b) Cumulative excess versus number of
reconfigurations for each policy family; dashed lines have constant regret
$\text{excess}+C N_T$.}
\label{fig:dynamic}
\end{figure}

\begin{table}[t]
\centering
\caption{Dynamic reconfiguration: number of reconfigurations, cumulative excess over the
oracle and regret $=$ excess $+\,C N_T$ (seconds; oracle total \DynOracle\,s).  Best regret per
family in bold.}
\label{tab:dynamic}
\small
\begin{tabular}{llrrr}
\toprule
Policy & parameter & $N_T$ & excess & regret\\
\midrule
static prior $\bfq$ & -- & 0 & 525.0 & 525.0\\
\midrule
periodic $k$ & 3 & 199 & 84.1 & 283.1\\
 & 10 & 59 & 83.1 & \textbf{142.1}\\
 & 30 & 19 & 129.8 & 148.8\\
 & 100 & 5 & 182.4 & 187.4\\
 & $\infty$ (once) & 0 & 1282.0 & 1282.0\\
\midrule
TV threshold $\vartheta$ & 0.05 & 9 & 107.2 & \textbf{116.2}\\
 & 0.1 & 3 & 117.5 & 120.5\\
 & 0.2 & 1 & 224.1 & 225.1\\
 & 0.3 & 1 & 224.1 & 225.1\\
 & 0.4 & 1 & 224.1 & 225.1\\
 & 0.6 & 1 & 224.1 & 225.1\\
\midrule
GLR--CUSUM $h$ & 300 & 599 & 79.7 & 678.7\\
 & 1000 & 599 & 79.7 & 678.7\\
 & 3000 & 39 & 79.1 & 118.1\\
 & 10000 & 16 & 87.2 & \textbf{103.2}\\
 & 30000 & 9 & 98.0 & 107.0\\
 & 100000 & 3 & 117.2 & 120.2\\
\bottomrule
\end{tabular}
\end{table}

\subsection{Multi-Tenant Fairness}

Three jobs share an OCS layer: a dense pipeline job ($4\times2$ pods), a skewed MoE job
($2\times4$, all-to-all 6\,GB) and a smaller MoE job ($2\times3$, 3\,GB), 22 pods with $r=8$.  The
reference times $J_k^{\rm ref}$ (each job alone with all of its own ports) are \MTref\,s.  We
solve the fair problem for $p\in\{1,2,4,8,\infty\}$ ($\lambda=10^{-3}$, $\tau_f=0.01$) with port budgets
$\mu=52$ and $\mu=36$, and compare with a static partition of the budget in proportion to the
number of pods, each job optimised alone in its partition.

\Cref{tab:mt,fig:multitenant} show the predicted behaviour of the fairness weights of
\cref{thm:kkt}.  As $p$ increases, the maximum slowdown decreases monotonically
(from 1.114 to 1.084 for $\mu=52$ and from 1.215 to 1.199 for $\mu=36$) while the mean slowdown
increases slightly (from 1.080 to 1.082 and from 1.198 to 1.199): the budget moves towards the
job that is currently worst off, at a small cost in aggregate efficiency.  All runs decreased
the free energy monotonically.  Sharing the budget dominates the static proportional
partition on both metrics: for $\mu=36$ the partition leaves the skewed MoE job at a slowdown
of 1.285 while over-provisioning the dense job (1.122), whereas every shared solution keeps all
slowdowns at or below 1.215.  The effect of $p$ is modest in this instance because the three
jobs respond similarly to additional circuits; larger differences arise when jobs differ more
strongly in their sensitivity to bandwidth.

\begin{figure}[t]
\centering
\includegraphics[width=\linewidth]{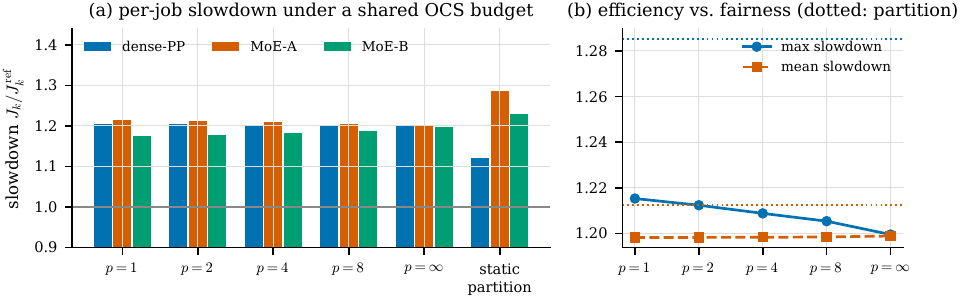}
\caption{Multi-tenant OCS budget ($\mu=36$).  (a) Slowdown of each job for increasing
fairness exponent $p$ and for a static proportional partition.  (b) Maximum and mean
slowdown; dotted lines show the static partition.}
\label{fig:multitenant}
\end{figure}

\begin{table}[t]
\centering
\caption{Multi-tenant results: per-job slowdowns $s_k$, maximum and mean slowdown, and
budget shares $\rho_k$ (dense / MoE-A / MoE-B).}
\label{tab:mt}
\small
\begin{tabular}{rlrrrrrl}
\toprule
$\mu$ & $p$ & $s_1$ & $s_2$ & $s_3$ & max & mean & shares $\rho$\\
\midrule
52 & 1 & 1.114 & 1.062 & 1.063 & 1.114 & 1.080 & 0.26 / 0.42 / 0.32\\
 & 2 & 1.110 & 1.065 & 1.065 & 1.110 & 1.080 & 0.27 / 0.42 / 0.31\\
 & 4 & 1.105 & 1.068 & 1.067 & 1.105 & 1.080 & 0.27 / 0.42 / 0.31\\
 & 8 & 1.099 & 1.071 & 1.071 & 1.099 & 1.080 & 0.28 / 0.41 / 0.31\\
 & $\infty$ & 1.084 & 1.081 & 1.080 & 1.084 & 1.082 & 0.30 / 0.40 / 0.30\\
 & static partition & 1.050 & 1.119 & 1.104 & 1.119 & 1.091 & 0.36 / 0.36 / 0.27\\
\midrule
36 & 1 & 1.205 & 1.215 & 1.174 & 1.215 & 1.198 & 0.26 / 0.42 / 0.32\\
 & 2 & 1.204 & 1.212 & 1.178 & 1.212 & 1.198 & 0.26 / 0.42 / 0.31\\
 & 4 & 1.203 & 1.209 & 1.183 & 1.209 & 1.198 & 0.26 / 0.43 / 0.31\\
 & 8 & 1.202 & 1.205 & 1.188 & 1.205 & 1.198 & 0.27 / 0.43 / 0.31\\
 & $\infty$ & 1.199 & 1.199 & 1.198 & 1.199 & 1.199 & 0.27 / 0.43 / 0.30\\
 & static partition & 1.122 & 1.285 & 1.230 & 1.285 & 1.212 & 0.36 / 0.36 / 0.27\\
\bottomrule
\end{tabular}
\end{table}

\section{Discussion}
\label{sec:discussion}

\paragraph{What the geometry contributes.}
The information-geometric structure is used for three concrete purposes.  (i) The KL
penalty to a symmetric prior gives a convex, strictly positive and symmetry-respecting
interpolation between static and adaptive wirings whose two ends are characterised exactly
(\cref{thm:limits}), and the Pythagorean inequality turns this into a quantitative gain
bound (\cref{thm:jeffreys}).  (ii) Entropic Bregman geometry makes the proximal step a
multiplicative update followed by a KL projection with a $|V|$-dimensional dual, and the
three-point property delivers linear convergence without Euclidean projections onto the
simplex.  (iii) Exponential families appear twice more: the Gibbs distribution over DAG paths
whose marginals are the criticalities, and the multinomial likelihood whose GLR statistic is
the KL indicator.

\paragraph{Limitations.}
\begin{itemize}
\item \emph{Model.}  Routing is single-hop on direct circuits in parallel with a static
  fabric, and concurrent phases are assumed not to contend (\cref{as:model}).  Multi-hop
  routing (as in ATRO \cite{atro2025}) would add a flow block coupled to the wiring through
  capacity constraints; the joint convexity of perspective functions
  $\kappa\,\phi(f/\kappa)$ makes this a natural extension, but its analysis is not done here.
\item \emph{Integrality.}  The theory concerns the continuous relaxation.  On our instances an
  exact MILP solved by HiGHS is fast and, for few ports per pod, clearly better than rounding
  plus local search (\cref{sec:exp_integer}); the continuous method is competitive when the
  number of ports per pod is large, where \cref{prop:rounding} predicts small rounding loss.
  We therefore do not claim that the method replaces integer programming.
\item \emph{Constants.}  $L$ in \cref{lem:smooth} and $\Gamma$ in \cref{lem:gamma} are
  worst-case and conservative; backtracking avoids needing $L$, but the regret bounds are
  loose in absolute terms and should be read as scaling laws.
\item \emph{Statistics.}  The false-alarm bound grows with the number of cells $M$; pooling
  cells (e.g.\ by pod pair) reduces $M$ at the cost of resolution.  The effect of estimating
  the reference distribution is not analysed.
\item \emph{Evaluation.}  All experiments are synthetic.  Validation on production traces
  and on real OCS hardware, including the cost and duration of reconfiguration, is required.
\end{itemize}

\section{Conclusion}
\label{sec:conclusion}

Optical circuit switching turns the wiring of an AI cluster from a fixed datum into a decision
variable.  We have argued that the appropriate language for this decision is the geometry of
the probability simplex.  In that language a single inverse temperature measures how far a
deployment departs from the symmetric Fat-Tree reference; the limits of this departure, the
gain it can buy and the circumstances in which it buys nothing admit exact statements; the
causal structure of a training iteration becomes an exponential family over DAG paths whose
marginals drive a convex and efficiently differentiable cost; and the decision to reconfigure
becomes a likelihood-ratio test whose consequences are quantified by a regret bound.  The same
language accommodates clusters shared by several jobs, where convex fairness is available
precisely in the range $\alpha\ge2$.

The experiments indicate both the value and the limits of this view.  The continuous theory
predicts the achievable iteration time accurately once pods have several optical ports, and
likelihood-ratio detection gives the best reconfiguration policy among those examined; exact
integer programming, however, remains preferable for small instances.  Extending the framework
to multi-hop routing and contention, and validating it on production traces and optical
hardware, are the natural next steps.

\appendix
\section{Deferred Proofs}
\label{app:proofs}

\begin{lemma}[Smoothed maximum of convex functions]
\label{lem:lse_comp}
Let $\phi_1,\dots,\phi_n$ be convex and $C^2$ on an open convex set $\mathcal O$ with
$\norm{\nabla\phi_i}_\infty\le G$ and $\bfv^\top\nabla^2\phi_i\bfv\le M\norm{\bfv}_1^2$.  Then
$f=\lse_\tau(\phi_1,\dots,\phi_n)$ is convex on $\mathcal O$,
$\norm{\nabla f}_\infty\le G$ and $\bfv^\top\nabla^2f\,\bfv\le(M+G^2/\tau)\norm{\bfv}_1^2$.
\end{lemma}

\begin{proof}
With softmax weights $\sigma_i\propto e^{\phi_i/\tau}$,
$\nabla f=\sum_i\sigma_i\nabla\phi_i$ and
\[
  \nabla^2f=\sum_i\sigma_i\nabla^2\phi_i+\tau^{-1}\Bigl(\sum_i\sigma_i\nabla\phi_i\nabla\phi_i^\top
  -\nabla f\nabla f^\top\Bigr).
\]  The bracket is a covariance matrix, positive semidefinite and
dominated by $\sum_i\sigma_i\nabla\phi_i\nabla\phi_i^\top$.  Hence $\nabla^2f\succeq0$ and
\[
  \bfv^\top\nabla^2f\bfv\le\sum_i\sigma_i\bigl(M\norm{\bfv}_1^2+\tau^{-1}(\nabla\phi_i^\top\bfv)^2\bigr)
  \le(M+G^2/\tau)\norm{\bfv}_1^2
\]
since $|\nabla\phi_i^\top\bfv|\le\norm{\nabla\phi_i}_\infty\norm{\bfv}_1$.
The gradient bound is immediate.
\end{proof}

\begin{proof}[Proof of \cref{lem:smooth}]
Work on the open set $\mathcal O=\{\bfw:\ b_0+Bw_e>b_0/2\ \forall e\}\supset\R^{\Eocs}_{\ge0}$ and
evaluate the bounds on $\R^{\Eocs}_{\ge0}$.  The transfer time
$t_{s,e}(\bfw)=d_{s,e}/(b_0+Bw_e)$ depends on one coordinate, is convex and decreasing, with
$|t'_{s,e}|=d_{s,e}B/\kappa_e^2\le G_s$ and $t''_{s,e}=2d_{s,e}B^2/\kappa_e^3\le2B^2\bar d_s/b_0^3$
for $\bfw\ge0$; thus $\bfv^\top\nabla^2t_{s,e}\bfv=t''_{s,e}v_e^2\le(2B^2\bar d_s/b_0^3)\norm{\bfv}_1^2$.
By \cref{lem:lse_comp}, $T_s^{\tau_s}$ is convex with $\norm{\nabla T^{\tau_s}_s}_\infty\le G_s$
and curvature constant $M_s$.  A path length $\ell_P$ is a sum of such terms and constants,
so it is convex with gradient bound $\sum_{v\in P}G_{s(v)}\le G_H$ and curvature constant
$\sum_{v\in P}M_{s(v)}\le M_H$.  A second application of \cref{lem:lse_comp} gives convexity
of $U$ and \eqref{eq:Lconst}.  Monotonicity holds because each $t_{s,e}$ is nonincreasing and
$\lse$ is nondecreasing in each argument.

For $\bfx,\bfy\in\Simp$, Taylor's theorem with integral remainder and \eqref{eq:Lconst} give
$U(\bfy)\le U(\bfx)+\inner{\nabla U(\bfx)}{\bfy-\bfx}+\frac L2\norm{\bfy-\bfx}_1^2$, and Pinsker's
inequality $\frac12\norm{\bfy-\bfx}_1^2\le\KL{\bfy}{\bfx}$ yields relative smoothness.  For the
approximation bound, $T_s\le T_s^{\tau_s}\le T_s+\tau_s\log n_s$, so
$\ell_P^{\rm hard}\le\ell_P\le\ell_P^{\rm hard}+H\tau_s\log\max_sn_s$, and
$\max_P\ell_P\le U\le\max_P\ell_P+\tau\log|\calP|$.
\end{proof}

\bibliographystyle{plainnat}

\end{document}